\documentclass[9pt, a4paper,twocolumn]{article} 

\usepackage[backend=biber,style=authoryear]{biblatex}
\usepackage{tabularx}
\usepackage{setspace}
\usepackage{subfloat}
\usepackage{mathtools}
\usepackage{amsmath,amsfonts,amssymb,amsthm,mathrsfs, bm}
\usepackage[utf8]{inputenc}
\usepackage[swedish, english]{babel}
\usepackage{csquotes}
\usepackage{subfig}
\usepackage{graphicx}
\usepackage{supertabular}                                             
\usepackage[table]{xcolor} 

\usepackage{graphicx}
\usepackage{epstopdf}
\usepackage{microtype}
\usepackage{multirow}
\usepackage{xcolor}
\usepackage{authblk}
\usepackage{datetime2}	
\usepackage[left=0.5in, right=0.5in, top=0.8in, bottom=0.8in]{geometry} 
\usepackage{setspace}	
\theoremstyle{plain}
\newtheorem{theorem}{Theorem}[section]
\newtheorem{lemma}{Lemma}[section]

\newtheorem{assumption}{Assumption}[section]

\usepackage{amsmath,amsfonts,amssymb,amsthm,mathrsfs, bm}
\usepackage{stmaryrd}

\usepackage{graphicx}
\usepackage{tikz}
\usepackage[utf8]{inputenc}
\usepackage{pgfplots} 
\usepackage{pgfgantt}
\usepackage{pdflscape}
\pgfplotsset{compat=newest} 
\pgfplotsset{plot coordinates/math parser=false}

\usepackage{endnotes}
\usepackage{commath}
\usepackage{gensymb}
\usepackage{mathtools}
\usepackage{tikz} 
\usetikzlibrary{calc,quotes,angles}
\usepackage{tkz-euclide}
\usetikzlibrary{automata,positioning} 

\usepackage{xcolor}
\usepackage{tikz,tkz-euclide,siunitx}
\usepackage{etoolbox} 

\usetikzlibrary{quotes,arrows.meta,angles,calc,shadings,positioning,babel}

\DeclareSymbolFont{cmletters}{OML}{cmm}{m}{it}
\DeclareMathSymbol{\epsilon}{\mathord}{cmletters}{"0F}

\newtheorem{remark}{Remark}

\usepackage{endnotes}
\usepackage{commath}
\usepackage{gensymb}

\makeatletter
\patchcmd{\tkz@DrawLine}{\begingroup}{\begingroup\makeatletter}{}{}
\makeatother

\allowdisplaybreaks

\DeclareMathOperator{\diag}{\mathrm{diag}}
\makeatletter
\newcommand\makebig[2]{%
  \@xp\newcommand\@xp*\csname#1\endcsname{\bBigg@{#2}}%
  \@xp\newcommand\@xp*\csname#1l\endcsname{\@xp\mathopen\csname#1\endcsname}%
  \@xp\newcommand\@xp*\csname#1r\endcsname{\@xp\mathclose\csname#1\endcsname}%
}
\makeatother

\providecommand*{\ped}[1]{%
\ensuremath{_\textnormal{#1}}}

\providecommand*{\eu}%
{\ensuremath{\mathrm{e}}}
\providecommand*{\GammaF}%
{\ensuremath{\mathrm{\Gamma}}}
\providecommand*{\BetaF}%
{\ensuremath{\mathrm{\Beta}}}
\usepackage{bbm}
\makebig{biggg} {3.0}
\makebig{Biggg} {3.5}
\makebig{bigggg}{4.0}
\makebig{Bigggg}{4.5}
\makebig{biggggg}{5.0}
\makebig{Biggggg}{5.5}
\usepackage{longtable}
\usepackage{xcolor}
\DeclareMathSymbol{\Gamma}{\mathalpha}{letters}{"00}
\DeclareMathSymbol{\Delta}{\mathalpha}{letters}{"01}
\DeclareMathSymbol{\Theta}{\mathalpha}{letters}{"02}
\DeclareMathSymbol{\Lambda}{\mathalpha}{letters}{"03}
\DeclareMathSymbol{\Xi}{\mathalpha}{letters}{"04}
\DeclareMathSymbol{\Pi}{\mathalpha}{letters}{"05}
\DeclareMathSymbol{\Sigma}{\mathalpha}{letters}{"06}
\DeclareMathSymbol{\Upsilon}{\mathalpha}{letters}{"07}
\DeclareMathSymbol{\Phi}{\mathalpha}{letters}{"08}
\DeclareMathSymbol{\Psi}{\mathalpha}{letters}{"09}
\DeclareMathSymbol{\Omega}{\mathalpha}{letters}{"0A}

\usepackage{hyperref}
\hypersetup{
    colorlinks = true,
    linkbordercolor = {white},
            linkcolor=blue,
            bookmarks=true,
            filecolor=blue,
            urlcolor=blue,
            citecolor=blue
}

\usepackage{cuted}%
\DeclareMathAlphabet{\altmathcal}{OMS}{cmsy}{m}{n}
\begin{document}
\title{Boundary adaptive observer design for semilinear hyperbolic\\ rolling contact ODE-PDE systems with uncertain friction}
\date{}
\author[a,b,c]{Luigi Romano\thanks{Corresponding author. Email: luigi.romano@liu.se.}}
\author[b]{Ole Morten Aamo}
\author[c]{Miroslav Krstić}
\author[a]{Jan Aslund}
\author[a]{Erik Frisk}
\affil[a]{\footnotesize{Department of Electrical Engineering, Linköping University, SE-581 83 Linköping, Sweden}}
\affil[b]{\footnotesize{Department of Engineering Cybernetics, Norwegian University of Science and Technology, O. S. Bragstads plass 2, NO-7034, Trondheim, Norway}}
\affil[c]{\footnotesize{Department of Mechanical and Aerospace Engineering, University of California San Diego, La Jolla, CA, 92093, USA}}
\maketitle
\begin{strip}
    \centering
    \begin{minipage}{.8\textwidth}
        \begin{abstract}
           This paper presents an adaptive observer design for semilinear hyperbolic rolling contact ODE-PDE systems with uncertain friction characteristics parameterized by a matrix of unknown coefficients appearing in the nonlinear (and possibly non-smooth) PDE source terms. Under appropriate assumptions of forward completeness and boundary sensing, an adaptive observer is synthesized to simultaneously estimate the lumped and distributed states, as well as the uncertain friction parameters, using only boundary measurements. The observer combines a finite-dimensional parameter estimator with an infinite-dimensional description of the state error dynamics, and achieves exponential convergence under persistent excitation. The effectiveness of the proposed design is demonstrated in simulation by considering a relevant example borrowed from road vehicle dynamics.
        \end{abstract}
        \hspace{0.5cm}

    \end{minipage}
\end{strip}


\section{Introduction}

Rolling contact systems with spatially-distributed friction effects arise in numerous mechanical applications, such as vehicle dynamics \parencite{TsiotrasConf,Tsiotras1,Tsiotras2,Deur1,Deur2,Mio}, rotating machinery \parencite{Frendo1,Frendo2}, and robotic locomotion \parencite{2D}. A common modeling approach involves interconnecting finite-dimensional ordinary differential equations (ODEs), which represent the rigid body dynamics, with hyperbolic partial differential equations (PDEs) that describe the behavior of the contact interface, in the presence of dry or lubricated friction \parencite{DistrLuGre,FrBD}. In many engineering applications, particularly in the automotive domain, accurately estimating system states and physical parameters is essential not only for real-time control but also for monitoring, diagnostics, and higher-level functions. For example, retrieving the tire-road friction coefficient is critical not only for the design of Advanced Driver Assistance Systems (ADAS), such as anti-lock braking systems (ABS) and electronic stability control (ESC) \parencite{Savaresi}, but also for road condition characterization, which may enhance safety, efficiency, and coordination of both individual vehicles and entire fleets \parencite{Shao,Shao2}.

Motivated by these considerations, this paper addresses the problem of joint state and parameter estimation for a class of semilinear, homodirectional ODE-PDE systems. These systems model rolling contact dynamics with distributed friction, where the PDE subsystem features nonlinear source terms governed by unknown parameters, and couples to the ODE dynamics through integral functionals of the distributed state and, possibly, boundary terms. The proposed observer reconstructs both the lumped and distributed states, as well as the friction-related parameters, using only boundary measurements. Importantly, in contrast to classical adaptive control frameworks where parameter estimation serves the primary goal of stabilizing or regulating the system, the focus here is on accurate estimation per se. This distinction reflects application scenarios in which reliable state and parameter reconstruction is valuable in its own right, independent of any direct control objective.

Adaptive observer design for PDE systems initially emerged within the broader framework of adaptive output-feedback control, particularly for parabolic PDEs \parencite{Smysh}. This foundational approach was later extended to first-order hyperbolic integro-differential equations with unknown functional coefficients in \textcite{Krstic1}, and to more general heterodirectional hyperbolic systems featuring both boundary and in-domain parametric uncertainties in \textcite{OleBook}. In \textcite{Ole0}, the disturbance rejection problem for coupled ODE-$2\times2$ hyperbolic systems was also addressed, albeit without parameter uncertainties. However, in many of these early developments, particularly those embedded in adaptive output-feedback controllers, convergence of the estimated parameters to their true values is not always guaranteed or even required. In contrast, a growing line of research has shifted focus toward adaptive observers aimed specifically at accurate estimation of both the system states and unknown parameters, independently of any control objective. Within this direction, which aligns with the scope of the present paper, adaptive observers have been developed for $2\times2$ hyperbolic systems in \textcite{Ghousein1}, wave equations with in-domain uncertainties in \textcite{Ben}, and general $n+1$ and $n+m$ hyperbolic PDEs in \textcite{Aamo,Aamo3}.

Concerning coupled hyperbolic ODE-PDE systems with unknown parameters, notable works include those of \textcite{Aamo2}, which consider a $2\times 2$ hyperbolic PDE coupled with an uncertain linear time-invariant (LTI) ODE, \textcite{Ghousein1}, which analyzes a class of linear homodirectional hyperbolic PDEs coupled with uncertain linear time-varying (LTV) ODEs, and \textcite{Traffic,Traffic2}, where the sources of uncertainty appear in the PDE subsystem. Most of these studies, however, are limited to linear, heterodirectional systems. With notable exceptions for the non-adaptive case (see, e.g., \textcite{Hasan}), the extension to semilinear equations, possibly homodirectional and with unknown parameters, also remains limited. The challenge lies in the nonlinearity of the PDE, the finite-dimensional nature of the parameter estimation problem, and the need to preserve well-posedness and convergence guarantees. The problem becomes even more delicate when considering ODE-PDE interconnections with nonlinear and non-differentiable coefficients, as is the case in the present paper, where friction-related terms introduce sources of non-smoothness.

Building upon recent advances in adaptive control and PDE estimation \parencite{Ole1,Ole2}, a novel observer architecture is proposed that combines a finite-dimensional adaptive law for parameter estimation based on filtered regression equations with an infinite-dimensional observer driven by boundary measurements. The approach establishes exponential convergence under mild assumptions on the system dynamics and excitation properties of the inputs. The mathematical analysis is based on standard Lyapunov arguments, similar to those utilized in \textcite{Coron}, in conjunction with uniform convergence results obtained for a regularized version of the non-smooth problem considered in the paper.
The proposed framework is validated considering a vehicle model with Dahl-type distributed tire friction, showing fast and robust convergence of both state and parameter estimates.


\subsection*{Notation}
In this paper, $\mathbb{R}$ denotes the set of real numbers; $\mathbb{R}_{>0}$ and $\mathbb{R}_{\geq 0}$ indicate the set of positive real numbers excluding and including zero, respectively; $\mathbb{N}_0$ indicates the set of integers including zero. 
The set of $n\times m$ matrices with values in $\mathbb{F}$ ($\mathbb{F} = \mathbb{R}$, $\mathbb{R}_{>0}$, or $\mathbb{R}_{\geq0}$) is denoted by $\mathbf{M}_{n\times m}(\mathbb{F})$ (abbreviated as $\mathbf{M}_{n}(\mathbb{F})$ whenever $m=n$). $\mathbf{GL}_n(\mathbb{R})$ and $\mathbf{Sym}_n(\mathbb{R})$ represents the groups of invertible and symmetric matrices, respectively, with values in $\mathbb{F}$; the identity matrix on $\mathbb{R}^n$ is indicated with $I_n$. A positive-definite matrix is noted as $\mathbf{M}_n(\mathbb{R}) \ni Q \succ 0$. The identity operator on a Banach space $\altmathcal{Z}$ is denoted by $I_{\altmathcal{Z}}$.
The standard Euclidean norm on $\mathbb{R}^n$ is indicated with $\norm{\cdot}_2$; operator norms are simply denoted by $\norm{\cdot}$.
$L^2((0,1);\mathbb{R}^n)$ denotes the Hilbert space of square-integrable functions on $(0,1)$ with values in $\mathbb{R}^n$, endowed with inner product $\langle \zeta_1, \zeta_2 \rangle_{L^2((0,1);\mathbb{R}^n)} = \int_0^1 \zeta_1^{\mathrm{T}}(\xi)\zeta_2(\xi) \dif \xi$ and induced norm $\norm{\zeta(\cdot)}_{L^2((0,1);\mathbb{R}^n)}$. The Hilbert space $H^1((0,1);\mathbb{R}^n)$ consists of functions $\zeta\in L^2((0,1);\mathbb{R}^n)$ whose weak derivative also belongs to $L^2((0,1);\mathbb{R}^n)$; it is naturally equipped with norm $\norm{\zeta(\cdot)}_{H^1((0,1);\mathbb{R}^n)}^2 \triangleq \norm{\zeta(\cdot)}_{L^2((0,1);\mathbb{R}^n)}^2 + \norm{\pd{\zeta(\cdot)}{\xi}}_{L^2((0,1);\mathbb{R}^n)}^2$.
Given a domain $\Omega$ with closure $\overline{\Omega}$, $L^p(\Omega;\altmathcal{Z})$ and $C^k(\overline{\Omega};\altmathcal{Z})$ ($p, k \in \{1, 2, \dots, \infty\}$) denote respectively the spaces of $L^p$-integrable functions and $k$-times continuously differentiable functions on $\overline{\Omega}$ with values in $\altmathcal{Z}$ (for $T = \infty$, the interval $[0,T]$ is identified with $\mathbb{R}_{\geq 0}$). Given two Hilbert spaces $\altmathcal{V}$ and $\altmathcal{W}$, $\mathscr{L}(\altmathcal{V};\altmathcal{W})$ denotes the spaces of linear operators from $\altmathcal{V}$ to $\altmathcal{W}$ (abbreviated $\mathscr{L}(\altmathcal{V})$ if $\altmathcal{V} = \altmathcal{W}$). Finally, following \textcite{Kato}, the group of operators on a Banach space $\altmathcal{Z}$ that are infinitesimal generators of a $C_0$-semigroup satisfying $\norm{T(t)} \leq \eu^{\omega t}$ is conventionally denoted by $\mathscr{G}(\altmathcal{Z}; 1, \omega)$.

\section{Problem statement}\label{sect:Problem}
This section is dedicated to introducing the considered ODE-PDE system, along with the main assumptions formulated about its structure and the available measurements.
\subsection{Structure of the considered ODE-PDE system}

In this paper, the following semilinear homodirectional hyperbolic ODE-PDE interconnection is considered:
\begin{subequations}\label{eq:originalSystems}
\begin{align}
\begin{split}
& \dot{X}(t) = A_1X(t) + G_1(\mathscr{K}_1z)(t) + G_1\Theta\Sigma\Bigl(v\bigl(X(t),U(t)\bigr)\Bigr)(\mathscr{K}_2z)(t)  \\
& \qquad \quad + G_1h_1\Bigl(v\bigl(X(t),U(t)\bigr)\Bigr), \quad   t \in (0,T),
\end{split} \label{eq:originalSystemsODE}\\
\begin{split}
& \dpd{z(\xi,t)}{t} + \Lambda \dpd{z(\xi,t)}{\xi} =\Theta \Sigma\Bigl(v\bigl(X(t),U(t)\bigr)\Bigr) z(\xi,t) \\
& \qquad \qquad \qquad \qquad \qquad \qquad +h_2\Bigl(v\bigl(X(t),U(t)\bigr)\Bigr), \\
& \qquad \qquad \qquad \qquad \qquad \qquad  (\xi,t) \in (0,1) \times (0,T),
\end{split} \label{eq:originalSystemsPDE} \\
& z(0,t) = 0, \quad t \in (0,T),\label{eq:originalSystemsBC}
\end{align}
\end{subequations}
where $X(t) \in \mathbb{R}^{n_X}$ indicates the lumped state vector, $z(\xi,t) \in \mathbb{R}^{n_z}$ denotes the distributed state vector, $U(t) \in \mathbb{R}^{n_U}$ represents the input to the PDE subsystem, $\mathbf{Sym}_{n_z}(\mathbb{R}) \ni \Theta = \diag\{\theta_1, \dots, \theta_{n_z}\} \succ 0$ is a diagonal matrix of unknown parameters, and the rigid relative velocity $v \in C^1(\mathbb{R}^{n_X+n_U};\mathbb{R}^{n_z})$ reads
\begin{align}\label{eq:relVel}
v(X,U) = A_2X + G_2U.
\end{align}
In~\eqref{eq:originalSystems} and~\eqref{eq:relVel}, the matrix $\mathbf{GL}_{n_z}(\mathbb{R})\cap \mathbf{Sym}_{n_z}(\mathbb{R}) \ni \Lambda = \diag\{\lambda_1, \dots, \lambda_{n_z}\} \succ 0$ collects the transport velocities, $\Sigma \in C^0(\mathbb{R}^{n_z};\mathbf{M}_{n_z}(\mathbb{R}))$ represents the nonlinear source matrix, the matrices $A_1 \in \mathbf{M}_{n_X}(\mathbb{R})$, $A_2 \in \mathbf{M}_{n_z\times n_X}(\mathbb{R})$, $G_1 \in  \mathbf{M}_{n_X\times n_z}(\mathbb{R})$, and $G_2\in  \mathbf{M}_{n_z\times n_U}(\mathbb{R})$ have constant coefficients, $h_1 \in C^0(\mathbb{R}^{n_z};\mathbb{R}^{n_z})$, and $h_2 \in C^0(\mathbb{R}^{n_z};\mathbb{R}^{n_z})$ is invertible.
Finally, the operators $(\mathscr{K}_1\zeta)$ and $(\mathscr{K}_2\zeta)$ satisfy $\mathscr{K}_1\in\mathscr{L}(H^1((0,1);\mathbb{R}^{n_z});\mathbb{R}^{n_z})$ and $\mathscr{K}_2 \in \mathscr{L}(L^2((0,1);\mathbb{R}^{n_z});\mathbb{R}^{n_z})$, and are given by
\begin{subequations}\label{eq;operatorsKs}
\begin{align}
(\mathscr{K}_1\zeta) &\triangleq   \int_0^1 K_1(\xi) \zeta(\xi) \dif \xi + K_2\zeta(1), \label{eq:operatorK_1}\\
(\mathscr{K}_2\zeta) &\triangleq   \int_0^1 K_3(\xi) \zeta(\xi) \dif \xi, \label{eq:operatorK_2}
\end{align}
\end{subequations}
with $K_1,K_3 \in C^0([0,1];\mathbf{M}_{n_z}(\mathbb{R}))$, and $K_2 \in \mathbf{M}_{n_z}(\mathbb{R})$. 


ODE-PDE interconnections of the type~\eqref{eq:originalSystems} typically describe rolling contact systems where the lumped states capture the rigid body dynamics, and the distributed ones describe the rolling contact motion in the presence of dry or lubricated friction. For instance, distributed friction models that can be put in the form~\eqref{eq:originalSystemsPDE} are the Dahl model, the LuGre model \parencite{DistrLuGre}, and certain variants of the FrBD model recently developed in \textcite{FrBD}.
From a mathematical perspective, the ODE-PDE interconnection~\eqref{eq:originalSystems} is (locally) well-posed. In particular, the Hilbert spaces $\altmathcal{X}\triangleq \mathbb{R}^{n_X}\times L^2((0,1);\mathbb{R}^{n_z})$ and $\altmathcal{Y}\triangleq \mathbb{R}^{n_X}\times H^1((0,1);\mathbb{R}^{n_z})$ are considered, equipped respectively with norms $\norm{(Z, \zeta(\cdot))}_{\altmathcal{X}}^2 \triangleq \norm{Z}_{2}^2 + \norm{\zeta(\cdot)}_{L^2((0,1);\mathbb{R}^{n_z})}^2$ and $\norm{(Z, \zeta(\cdot))}_{\altmathcal{Y}}^2 \triangleq \norm{Z}_{2}^2 + \norm{\zeta(\cdot)}_{H^1((0,1);\mathbb{R}^{n_z})}^2$, along with the domain $\mathscr{D}(\mathscr{A}) \triangleq \{ (Z,\zeta)\in \altmathcal{Y} \mathrel{|} \zeta(0) = 0\}$
With the above definitions, local well-posedness is enounced by Theorem~\ref{thm:mild} below.

\begin{theorem}[Local existence and uniqueness of solutions]\label{thm:mild}
Suppose that $\Sigma \in C^0(\mathbb{R}^{n_z};\mathbf{M}_{n_z}(\mathbb{R}))$ and $h_1, h_2\in C^0(\mathbb{R}^{n_z};\mathbb{R}^{n_z})$ are locally Lipschitz continuous, and $U \in C^0([0,T];\mathbb{R}^{n_U})$. Then, for all initial conditions (ICs) $(X_0,z_0) \triangleq (X(0),z(\cdot,0)) \in \altmathcal{X}$, there exists $ t\ped{max} \leq \infty$ such that the ODE-PDE system~\eqref{eq:originalSystems} admits a unique \emph{mild solution} $(X,z) \in C^0([0,t\ped{max});\altmathcal{X})$. Additionally, if $\Sigma \in C^1(\mathbb{R}^{n_z};\mathbf{M}_{n_z}(\mathbb{R}))$, $h_1, h_2\in C^1(\mathbb{R}^{n_z};\mathbb{R}^{n_z})$, and $U \in C^1([0,T];\mathbb{R}^{n_U})$, for all ICs $(X_0,z_0) \in \mathscr{D}(\mathscr{A})$, the solution is \emph{classical}, i.e., $(X,z) \in C^1([0,t\ped{max});\altmathcal{X}) \cap C^0([0,t\ped{max}); \mathscr{D}(\mathscr{A}))$.
\begin{proof}
The proof is similar to that of Theorems 3.1 and 3.2 in \textcite{SemilinearMio}, and hence omitted for brevity.
\end{proof}
\end{theorem}
Theorem~\ref{thm:mild} merely asserts local existence and uniqueness results for the solutions of~\eqref{eq:originalSystems}. In the following, it is explicitly assumed that the input $U\in L^\infty(\mathbb{R}_{\geq 0};\mathbb{R}^{n_U})$ is such that the existence of such solutions is global, so that $t\ped{max} = \infty$, and moreover that the solution of~\eqref{eq:originalSystems} remains uniformly bounded in time in an appropriate spatial norm. These assumptions are better formalized below.

\begin{assumption}[Forward completeness and boundedness]\label{ass:fwd}
The ODE-PDE system~\eqref{eq:originalSystems} is forward complete, that is, $t\ped{max} = \infty$ in Theorem~\ref{thm:mild}. Moreover, $(X,z) \in L^\infty(\mathbb{R}_{\geq 0};\altmathcal{X})$ for mild solutions, $(X,z) \in L^\infty(\mathbb{R}_{\geq 0};\altmathcal{Y})$ for classical solutions, and $U \in L^\infty(\mathbb{R}_{\geq 0}; \mathbb{R}^{n_U})$.
\end{assumption}

Concerning specifically the global well-posedness (forward-completeness) requirement of Assumption~\ref{ass:fwd}, additional details are discussed below. 
\begin{remark}[Structural conditions for global well-posedness]\label{rem:global-wp}
Global well-posedness of the ODE-PDE system~\eqref{eq:originalSystems} requires
specific structural properties of the nonlinear functions $\Sigma(\cdot)$,
$h_1(\cdot)$, and $h_2(\cdot)$. Following the framework established in~\textcite[Theorem~3.3]{SemilinearMio},
the following conditions are sufficient:
\begin{enumerate}
\item[(i)] The matrix $K_2 = 0$ in~\eqref{eq:operatorK_1}.
\item[(ii)] \emph{Linear growth of $h_1(\cdot),h_2(\cdot)$:} there exist constants
$L_{h_1},L_{h_2},b_1,b_2\in \mathbb{R}_{\geq 0}$ such that, for all $v\in\mathbb{R}^{n_z}$,
\begin{equation}
\norm{h_1(v)}_2 \leq L_{h_1}\norm{v}_2 + b_1, 
\quad
\norm{h_2(v)}_2 \le L_{h_2}\norm{v}_2 + b_2 .
\end{equation}
\item[(iii)] \emph{Boundedness or dissipativity of $\Sigma(\cdot)$:} either
\begin{enumerate}
\item[(H.2)] $\Sigma(\cdot)$ is uniformly bounded,
i.e.,\ $\norm{\Sigma(y)}\leq M_\Sigma$ for some $M_\Sigma \in \mathbb{R}_{\geq 0}$ and all $y\in\mathbb{R}^{n_z}$, or
\item[(H.1)] there exists a diagonal matrix-valued function
$P\in C^0([0,1];\mathbf{Sym}_{n_z}(\mathbb{R}))$, $P(\xi)\succ 0$, such that
\begin{equation}
\int_0^1 \zeta^{\mathrm{T}}(\xi)P(\xi)\Sigma(y)\zeta(\xi) \dif \xi \leq 0,
\end{equation}
for all $(y,\zeta) \in \mathbb{R}^{n_z}\times L^2((0,1);\mathbb{R}^{n_z})$.
\end{enumerate}
\end{enumerate}
The conditions above are satisfied by the distributed Dahl and FrBD friction models, but not by LuGre, due to its heuristic derivation~\textcite{SemilinearMio}.
\end{remark}

Under Assumption~\ref{ass:fwd}, the scope of this paper consists of reconstructing the parameters contained in the matrix $\Theta$, and the lumped and distributed states, using only boundary measurements from the PDE subsystem~\eqref{eq:originalSystemsPDE}, as clarified in the next Section~\ref{sect:measurements}. The reasons for incorporating all uncertainties into the matrix $\Theta$ are summarized below. 
\begin{remark}
From a physical viewpoint, the parameters in $\Sigma(v)$ are geometric or elastic quantities that can be identified off-line through static tests, whereas, under standard friction-model parametrizations \parencite{CanudasFr2,CanudasFr1}, $\Theta$ captures the friction coefficient, which depends on road conditions and cannot be measured directly during operation. Moreover, from a mathematical viewpoint, concentrating all uncertainty in $\Theta$ leads to linear measurement equations and certainty-equivalence adaptive laws. Hence, the proposed parametrization is consistent with the adaptive friction-estimation literature \parencite{Shao,Shao2,CanudasFr2,CanudasFr1} and simplifies the analysis compared with the joint estimation of geometric and elastic parameters.
\end{remark} 
In the following, convergence results for the error dynamics will be asserted in two different norms ($\norm{\cdot}_{\altmathcal{X}}$ and $\norm{\cdot}_{\altmathcal{Y}}$, respectively), depending on the smoothness of the matrix-valued function $\Sigma(\cdot)$ figuring in~\eqref{eq:originalSystems}. Indeed, for $\Sigma \in C^0(\mathbb{R}^{n_z};\mathbf{M}_{n_z}(\mathbb{R}))$ and $h_1, h_2\in C^0(\mathbb{R}^{n_z};\mathbb{R}^{n_z})$ locally Lipschitz, Theorem~\ref{thm:mild} only guarantees the existence and uniqueness of mild solutions, which motivates studying the convergence of the observer error dynamics in the norm $\norm{\cdot}_{\altmathcal{X}}$. Conversely, $\Sigma \in C^1(\mathbb{R}^{n_z};\mathbf{M}_{n_z}(\mathbb{R}))$ and $h_1, h_2\in C^1(\mathbb{R}^{n_z};\mathbb{R}^{n_z})$ permit recovering classical solutions and consequently stronger estimates for the observer error dynamics, whose convergence may be studied in the norm $\norm{\cdot}_{\altmathcal{Y}}$. Despite the intrinsic non-smooth nature of friction, the case with $\Sigma \in C^1(\mathbb{R}^{n_z};\mathbf{M}_{n_z}(\mathbb{R}))$ and $h_1, h_2\in C^1(\mathbb{R}^{n_z};\mathbb{R}^{n_z})$ is considered for completeness.


\subsection{Assumptions and considerations\\ on the available measurements}\label{sect:measurements}
Concerning the measured signals, two sets of measurements are supposed to be available: the first is needed for the design of a non-adaptive observer, whereas the second facilitates the synthesis of the adaptive one. Both types of signals considered in this paper may be acquired, as is typically the case, using accelerometers or strain gauges mounted on the rolling contact components. As these sensors rotate with the mechanical elements, they capture state variations in the Lagrangian reference frame, yielding direct measurements of true accelerations and velocities. In contrast, the PDE~\eqref{eq:originalSystemsPDE} is formulated in the Eulerian framework. The relationship between the Lagrangian and Eulerian time derivatives is given by
\begin{align}\label{eq:LagEul}
\dod{z(\xi,t)}{t} = \dpd{z(\xi,t)}{t} + \Lambda\dpd{z(\xi,t)}{\xi}.
\end{align}

Starting with~\eqref{eq:LagEul}, the first set of measurements is defined as
\begin{align}\label{eq:Y1}
Y_1(t) & = \dod{z(0,t)}{t} = \Lambda\dpd{z(0,t)}{\xi} =h_2\Bigl(v\bigl(X(t),U(t)\bigr)\Bigr),
\end{align}
giving
\begin{align}\label{eq:Y1upsilon}
v\bigl(X(t),U(t)\bigr) = h_2^{-1}\bigl(Y_1(t)\bigr).
\end{align}
Equations~\eqref{eq:Y1} and~\eqref{eq:Y1upsilon} permits reconstructing the rigid relative velocity directly from available measurements.
Accordingly, the following assumption is formulated.
\begin{assumption}[Detectability]\label{ass:rightInv}
The pair $(A_1, A_2)$ is detectable.
\end{assumption}
The second set of measurements is instead defined as
\begin{align}\label{eq:Y2}
\begin{split}
Y_2(t) & = \dod[2]{z(0,t)}{t} \\
&= \Lambda\Theta \Sigma\Bigl(v\bigl(X(t),U(t)\bigr)\Bigr) \dpd{z(0,t)}{\xi} + \dot{h}_2\Bigl(v\bigl(X(t),U(t)\bigr)\Bigr) \\
& = \Theta\Lambda \Sigma\Bigl(h_2^{-1}\bigl(Y_1(t)\bigr)\Bigr)\Lambda^{-1}Y_1(t)+\dot{Y}_1(t),
\end{split}
\end{align}
where the fact that $\Lambda$ and $\Theta$ commute has been used.

Before moving to the design of an adaptive observer, some considerations are formalized in Remark~\ref{remark:meas} below.
\begin{remark}\label{remark:meas}
As indicated by~\eqref{eq:Y1} and~\eqref{eq:Y2}, acquiring the full set of $2n_z$ measurements may initially seem like a stringent requirement. However, in most practical cases, certain components of $X(t)$ are readily accessible, enabling the recovery of the rigid relative velocity $v(X(t),U(t))$ by observing only a single component of the PDE subsystem. Moreover, the number of unknown parameters in the matrix $\Theta$ can often be reduced, which in turn decreases the dimensionality of $Y_2(t)$. Nevertheless, to present a simple, general, and compact framework for joint parameter and state estimation focused exclusively on boundary sensing, this paper assumes that the full set of $2n_z$ measurements is available.
\end{remark}

\section{Adaptive observer design}\label{sect:observer0}
The present section addresses the synthesis of an adaptive observer design for the ODE-PDE system~\eqref{eq:originalSystems}, assuming that the friction parameters contained in the matrix $\Theta$ are uncertain.

\subsection{Parameter identifiers and adaptive laws}\label{sect:parameter}
The first step consists of deriving a suitable adaptive law to estimate the unknown parameters contained in the matrix $\Theta$. The ODE-based approach adopted in this paper is partly inspired from \textcite{Ole1}, and makes only use of finite-dimensional tools. In particular, defining the vector $\mathbb{R}^{n_z}\ni 1_{n_z}\triangleq [1\; \dots \; 1]^{\mathrm{T}}$, and introducing the variables
\begin{align}\label{eq:sssss3}
Z_1(t) & \triangleq  1_{n_z}^{\mathrm{T}}Y_1(t), \quad Z_2(t) \triangleq  1_{n_z}^{\mathrm{T}}Y_2(t), \quad \theta  \triangleq  \Theta 1_{n_z},\nonumber \\
 \quad \phi(t) & \triangleq -\Lambda\Sigma\Bigl(h_2^{-1}\bigl(Y_1(t)\bigr)\Bigr)\Lambda^{-1}Y_1(t),
\end{align}
a linear parametric model may be derived as
\begin{align}\label{eq:parametric1}
\dot{Z}_1(t) = \theta^{\mathrm{T}}\phi(t) + Z_2(t).
\end{align}
The following filters may now be designed:
\begin{subequations}\label{eq:filters}
\begin{align}
\dot{\zeta}_1(t) & = -\varrho\bigl(\zeta_1(t)-Z_1(t)\bigr), \\
\dot{\zeta}_2(t) & = -\varrho\zeta_2(t) + Z_2(t), \\
\dot{\varphi}(t) & = -\varrho\varphi(t) + \phi(t), \quad t \in (0,T),\label{eq:varphi}
\end{align}
\end{subequations}
for some gain $\varrho \in \mathbb{R}_{>0}$. Accordingly, the non-adaptive estimate $\bar{Z}_1(t) \in \mathbb{R}$ of $Z_1(t)$ may be generated from
\begin{align}\label{eq:parametric1}
\bar{Z}_1(t) = \theta^{\mathrm{T}}\varphi(t) + \zeta(t),
\end{align}
with $\mathbb{R} \ni \zeta (t) \triangleq \zeta_1(t) + \zeta_2(t)$. It should be noted that $(X,z) \in C^0([0,T];\altmathcal{X})$ implies that $\bar{Z}_1, \zeta \in C^0([0,T];\mathbb{R})$ and $\varphi \in C^0([0,T];\mathbb{R}^{n_z})$, ensuring that the signals in~\eqref{eq:parametric1} are well-defined for both mild and classical solutions of~\eqref{eq:originalSystems}. Formally, the dynamics of the non-adaptive state error $\mathbb{R}\ni \tilde{Z}_1(t) \triangleq Z_1(t)-\bar{Z}_1(t)$ may be easily deduced to obey
\begin{align}
\dot{\tilde{Z}}_1(t) & = -\varrho\tilde{Z}_1(t), \quad t \in (0,T).
\end{align}
At this point, any standard adaptive law may be employed to estimate the parameters in $\theta$, as explained in \textcite{Ioannou}. Below, the gradient law with integral cost function is stated. In the following, the estimate of $\theta$ is denoted as $\hat{\theta}(t) \in \mathbb{R}^{n_z}$, whereas the parameter estimation error as $\mathbb{R}^{n_z} \ni \tilde{\theta}(t) \triangleq \theta-\hat{\theta}(t)$.

\begin{theorem}\label{thm:3.1}
Suppose that $U \in C^0([0,T];\mathbb{R}^{n_U})$, and consider the ODE-PDE interconnection~\eqref{eq:originalSystems}, with measurements~\eqref{eq:Y1} and~\eqref{eq:Y2}, and filters~\eqref{eq:filters}. If Assumption~\ref{ass:fwd} holds, the adaptive law
\begin{align}\label{eq:dottheta}
\dot{\hat{\theta}}(t)  = -\Gamma\bigl(R(t)\hat{\theta}(t) + Q(t)\bigr), \quad t \in (0,T),
\end{align}
where $\mathbf{Sym}_{n_z}(\mathbb{R})\ni \Gamma \succ 0$, and
\begin{subequations}\label{eq:RQ}
\begin{align}
\dot{R}(t) & = -\psi R(t) + \dfrac{\varphi(t)\varphi^{\mathrm{T}}(t)}{1 + \norm{\varphi(t)}_2^2}, \\
\dot{Q}(t) & = -\psi Q(t) - \dfrac{\bigl(Z_1(t)-\zeta(t)\bigr)\varphi (t)}{1 + \norm{\varphi(t)}_2^2}, \quad t \in (0,T),
\end{align}
\end{subequations}
with $R(t) \in \mathbf{M}_{n_z}(\mathbb{R})$, $Q(t) \in \mathbb{R}^{n_z}$, and $\psi \in \mathbb{R}_{>0}$, ensures that
\begin{subequations}
\begin{align}
\tilde{\theta}  & \in L^{\infty}(\mathbb{R}_{\geq 0};\mathbb{R}^{n_z}),\label{eq:thetaCOnd1} \\
\tilde{Z}_1 & \in L^{\infty}(\mathbb{R}_{\geq 0};\mathbb{R})\times L^2(\mathbb{R}_{> 0};\mathbb{R}), \\
\dot{\hat{\theta}} &  \in L^{\infty}(\mathbb{R}_{\geq 0};\mathbb{R}^{n_z})\times L^2(\mathbb{R}_{> 0};\mathbb{R}^{n_z}), \\
\lim_{t \to \infty} &\norm{\dot{\hat{\theta}}(t)}_2 = 0.
\end{align}
\end{subequations}
Moreover, if $\varphi \in L^\infty(\mathbb{R}_{\geq 0};\mathbb{R}^{n_z})$ in~\eqref{eq:varphi} is \emph{persistently exciting} (PE), then $\norm{\tilde{\theta}(t)}_2 \to 0$ exponentially fast.
\begin{proof}
See \textcite{Ioannou}.
\end{proof}
\end{theorem}

An estimate of the matrix $\Theta$ may thus be constructed as $\mathbf{Sym}_{n_z}(\mathbb{R}) \ni \hat{\Theta}(t) = \diag\{\hat{\theta}_1(t), \dots, \hat{\theta}_{n_z}(t)\}$, so that $\mathbf{Sym}_{n_z}(\mathbb{R})\ni \tilde{\Theta}(t) \triangleq \Theta - \hat{\Theta}(t) = \diag\{\tilde{\theta}_1(t), \dots, \tilde{\theta}_{n_z}(t)\}$. In this context, it should also be emphasized that~\eqref{eq:dottheta} and~\eqref{eq:RQ}, in conjunction with~\eqref{eq:sssss3} and~\eqref{eq:filters}, imply that $\hat{\Theta},\tilde{\Theta} \in C^1(\mathbb{R}_{\geq 0};\mathbf{Sym}_{n_z}(\mathbb{R}))$.
Obviously $\norm{\hat{\Theta}(t)} \leq \norm{\hat{\theta}(t)}_2$ and $\norm{\tilde{\Theta}(t)} \leq \norm{\tilde{\theta}(t)}_2$, and hence from~\eqref{eq:thetaCOnd1} it immediately follows that
\begin{align}
\norm{\hat{\Theta}}, \norm{\tilde{\Theta}}\in L^{\infty}(\mathbb{R}_{\geq 0};\mathbf{Sym}_{n_z}(\mathbb{R})).
\end{align}
Additionally, $\norm{\tilde{\Theta}(t)} \to 0$ exponentially fast whenever $\varphi \in L^\infty(\mathbb{R}_{\geq 0};\mathbb{R}^{n_z})$ in~\eqref{eq:varphi} is PE. Owing to these premises, it is possible to proceed with the synthesis of an adaptive observer to estimate the lumped and distributed states. This is the objective of the next Section~\ref{sect:observer}.

\subsection{Adaptive observer}\label{sect:observer}
Denoting the estimates of $X(t)$, $z(\xi,t)$, and $Y_1(t)$ respectively as $\hat{X}(t) \in \mathbb{R}^{n_X}$, $\hat{z}(\xi,t) \in \mathbb{R}^{n_z}$, and $\hat{Y}_1(t) \in \mathbb{R}^{n_z}$, the following observer structure is proposed:
\begin{subequations}\label{eq:observer}
\begin{align}
\begin{split}
& \dot{\hat{X}}(t) = A_1\hat{X}(t) + G_1(\mathscr{K}\hat{z})(t) \\
& \qquad \quad + G_1\hat{\Theta}(t)\Sigma\Bigl(h_2^{-1}\bigl(Y_1(t)\bigr)\Bigr)(\mathscr{K}_2\hat{z})(t) +  G_1h_1\Bigl(h_2^{-1}\bigl(Y_1(t)\bigr)\Bigr)  \\
& \qquad \quad-L_1\Bigl(h_2^{-1}\bigl(Y_1(t)\bigr)-h_2^{-1}\bigl(\hat{Y}_1(t)\bigr)\Bigr), \quad t \in (0,T),
\end{split} \label{eq:SystemsODEObs}\\
\begin{split}
& \dpd{\hat{z}(\xi,t)}{t} + \Lambda \dpd{\hat{z}(\xi,t)}{\xi} =\hat{\Theta}(t)\Sigma\Bigl(h_2^{-1}\bigl(Y_1(t)\bigr)\Bigr)\hat{z}(\xi,t) + Y_1(t)\\
& \qquad \qquad \qquad  \quad (\xi,t) \in (0,1) \times (0,T),
\end{split} \label{eq:SystemsPDEObs} \\
& \hat{z}(0,t) = 0, \quad t \in (0,T),\label{eq:SystemsBCObs}
\end{align}
\end{subequations}
where $L_1 \in \mathbf{M}_{n_X\times n_z}(\mathbb{R})$ is a matrix with constant coefficients. The estimated output reads
\begin{align}
\hat{Y}_1(t) = h_2\Bigl(v\bigl(\hat{X}(t),U(t)\bigr)\Bigr) = h_2\bigl(A_2\hat{X}(t) + G_2U(t)\bigr).
\end{align}
Consequently, defining the error variables as $\mathbb{R}^{n_X} \ni \tilde{X}(t) \triangleq X(t)-\hat{X}(t)$ and $\mathbb{R}^{n_z}\ni \tilde{z}(\xi,t) \triangleq z(\xi,t)-\hat{z}(\xi,t)$, and subtracting~\eqref{eq:observer} from~\eqref{eq:originalSystems}, the observer error dynamics is governed by
\begin{subequations}\label{eq:observerERR}
\begin{align}
\begin{split}
& \dot{\tilde{X}}(t) = A_1\tilde{X}(t) + G_1(\mathscr{K}\tilde{z})(t) \\
&\qquad \quad + G_1\hat{\Theta}(t)\Sigma\Bigl(v\bigl(X(t),U(t)\bigr)\Bigr)(\mathscr{K}_2\tilde{z})(t)  \\
 &\qquad \quad  + G_1\tilde{\Theta}(t)\Sigma\Bigl(v\bigl(X(t),U(t)\bigr)\Bigr)(\mathscr{K}_2z)(t)  \\
&\qquad \quad  +L_1\Bigl(h_2^{-1}\bigl(Y_1(t)\bigr)-h_2^{-1}\bigl(\hat{Y}_1(t)\bigr)\Bigr), \quad t \in (0,T),
\end{split} \label{eq:SystemsODEObsErr}\\
\begin{split}
& \dpd{\tilde{z}(\xi,t)}{t} + \Lambda \dpd{\tilde{z}(\xi,t)}{\xi} =\hat{\Theta}(t)\Sigma\Bigl(v\bigl(X(t),U(t)\bigr)\Bigr)\tilde{z}(\xi,t) \\
&\qquad \qquad +\tilde{\Theta}(t)\Sigma\Bigl(v\bigl(X(t),U(t)\bigr)\Bigr)z(\xi,t),\\
& \qquad \qquad (\xi,t) \in (0,1) \times (0,T),
\end{split} \label{eq:SystemsPDEObsErr} \\
& \tilde{z}(0,t) = 0, \quad t \in (0,T),\label{eq:SystemsBCObsErr}
\end{align}
\end{subequations}
where
\begin{align}
h_2^{-1}\bigl(Y_1(t)\bigr)-h_2^{-1}\bigl(\hat{Y}_1(t)\bigr) = v\bigl(\tilde{X}(t),0\bigr) = A_2\tilde{X}(t).
\end{align}
Recalling Assumption~\ref{ass:rightInv}, $L_1$ may be chosen such that $\mathbf{M}_{n_X}(\mathbb{R}) \ni \bar{A}_1 \triangleq A_1 + L_1A_2$ is Hurwitz. Thus, the observer error dynamics~\eqref{eq:observerERR} becomes
\begin{subequations}\label{eq:obSGAII}
\begin{align}
\begin{split}
& \dot{\tilde{X}}(t) = \bar{A}_1\tilde{X}(t) + G_1(\mathscr{K}_1\tilde{z})(t) + F_1(t)(\mathscr{K}_2\tilde{z})(t)   \\
& \qquad\quad+ f_1(t),\quad  t \in (0,T),
\end{split} \label{eq:SystemsODEObsErrL}\\
\begin{split}
& \dpd{\tilde{z}(\xi,t)}{t} + \Lambda \dpd{\tilde{z}(\xi,t)}{\xi} = F_2(t)\tilde{z}(\xi,t) +f_2(\xi,t),\\
&\qquad \qquad \qquad \qquad \qquad \qquad  (\xi,t) \in (0,1) \times (0,T),
\end{split} \label{eq:originalSystemsPDEObsErrL} \\
& \tilde{z}(0,t) = 0, \quad t \in (0,T),\label{eq:originalSystemsBCObsErrL}
\end{align}
\end{subequations}
where
\begin{subequations}\label{eq:Fns}
\begin{align}
F_1(t) &\triangleq G_1\hat{\Theta}(t)\Sigma\Bigl(v\bigl(X(t),U(t)\bigr)\Bigr), \\
F_2(t) &\triangleq \hat{\Theta}(t)\Sigma\Bigl(v\bigl(X(t),U(t)\bigr)\Bigr), 
\end{align}
\end{subequations}
and
\begin{subequations}
\begin{align}
f_1(t) & \triangleq G_1\tilde{\Theta}(t)\Sigma\Bigl(v\bigl(X(t),U(t)\bigr)\Bigr)(\mathscr{K}_2z)(t), \\
f_2(\xi,t) & \triangleq \tilde{\Theta}(t)\Sigma\Bigl(v\bigl(X(t),U(t)\bigr)\Bigr) z(\xi,t).
\end{align}
\end{subequations}
Assumption~\ref{ass:fwd} ensures that $(X,z) \in C^0(\mathbb{R}_{\geq 0};\altmathcal{X})\cap L^\infty(\mathbb{R}_{\geq 0};\altmathcal{X})$ and $U \in C^0(\mathbb{R}_{\geq 0}; \mathbb{R}^{n_U})\cap L^\infty(\mathbb{R}_{\geq 0}; \mathbb{R}^{n_U})$, which, combined with $\hat{\Theta},\tilde{\Theta} \in C^1(\mathbb{R}_{\geq 0};\mathbf{Sym}_{n_z}(\mathbb{R}))\cap L^\infty(\mathbb{R}_{\geq 0};\mathbf{Sym}_{n_z}(\mathbb{R}))$, implies that $F_1 \in C^0(\mathbb{R}_{\geq 0};\mathbf{M}_{n_X\times n_z}(\mathbb{R})) \cap L^\infty(\mathbb{R}_{\geq 0};\mathbf{M}_{n_X\times n_z}(\mathbb{R}))$, $F_2  \in C^0(\mathbb{R}_{\geq 0};\mathbf{M}_{n_z}(\mathbb{R}))\cap L^\infty(\mathbb{R}_{\geq 0};\mathbf{M}_{n_z}(\mathbb{R}))$, $f_1 \in C^0(\mathbb{R}_{\geq 0};\mathbb{R}^{n_X})\cap L^\infty(\mathbb{R}_{\geq 0};\mathbb{R}^{n_X})$, and $f_2 \in C^0(\mathbb{R}_{\geq 0};L^2((0,1);\mathbb{R}^{n_z}))\cap L^\infty(\mathbb{R}_{\geq 0};L^2((0,1);\mathbb{R}^{n_z}))$. Thus, similar arguments as those in the proof of Theorem~\ref{thm:mild} yield the existence of a unique (global) mild solution $(\tilde{X},\tilde{z}) \in C^0(\mathbb{R}_{\geq 0};\altmathcal{X})$ to~\eqref{eq:obSGAII} for all ICs $(\tilde{X}_0,\tilde{z}_0) \triangleq (\tilde{X}(0),\tilde{z}(\cdot,0)) \in \altmathcal{X}$. However, since $\Sigma(\cdot)$ is only supposed to be continuous in its argument, classical solutions to the ODE-PDE system~\eqref{eq:obSGAII} are not guaranteed by Theorem~\ref{thm:mild}.

Therefore, for $n\in\mathbb{N}_0$, the following regularized version of the ODE-PDE system~\eqref{eq:obSGAII} is also considered:  
\begin{subequations}\label{eq:obSGAIIn}
\begin{align}
\begin{split}
& \dot{\tilde{X}}^n(t) = \bar{A}_1\tilde{X}^n(t) + G_1(\mathscr{K}_1\tilde{z}^n)(t) + F_1^n(t)(\mathscr{K}_2\tilde{z}^n)(t)   \\
& \qquad\quad+ f_1^n(t),\quad  t \in (0,T),
\end{split} \label{eq:SystemsODEObsErrLn}\\
\begin{split}
& \dpd{\tilde{z}^n(\xi,t)}{t} + \Lambda\dpd{\tilde{z}^n(\xi,t)}{\xi} = F_2^n(t)\tilde{z}^n(\xi,t) +f_2^n(\xi,t),\\
&\qquad \qquad \qquad \qquad \qquad \qquad  (\xi,t) \in (0,1) \times (0,T),
\end{split} \label{eq:originalSystemsPDEObsErrLn} \\
& \tilde{z}^n(0,t) = 0, \quad t \in (0,T),\label{eq:originalSystemsBCObsErrLn}
\end{align}
\end{subequations}
where $F_1^n \in C^1(\mathbb{R}_{\geq 0};\mathbf{M}_{n_X\times n_z}(\mathbb{R})) \cap L^\infty(\mathbb{R}_{\geq 0};\mathbf{M}_{n_X\times n_z}(\mathbb{R}))$, $F_2^n  \in C^1(\mathbb{R}_{\geq 0};\mathbf{M}_{n_z}(\mathbb{R}))\cap L^\infty(\mathbb{R}_{\geq 0};\mathbf{M}_{n_z}(\mathbb{R}))$, $f_1^n \in C^1(\mathbb{R}_{\geq 0};\mathbb{R}^{n_X})\cap L^\infty(\mathbb{R}_{\geq 0};\mathbb{R}^{n_X})$, and $f_2^n \in C^1(\mathbb{R}_{\geq 0};L^2((0,1);\mathbb{R}^{n_z}))\cap L^\infty(\mathbb{R}_{\geq 0};L^2((0,1);\mathbb{R}^{n_z}))$ satisfy
\begin{subequations}\label{eq:Xnnnnnnnn}
\begin{align}
F_1^n & \xrightarrow[n \to \infty]{} F_1, \quad \textnormal{in}\; L^\infty(\mathbb{R}_{\geq 0};\mathbf{M}_{n_X\times n_z}(\mathbb{R})), \\
F_2^n & \xrightarrow[n \to \infty]{} F_2, \quad \textnormal{in}\; L^\infty(\mathbb{R}_{\geq 0};\mathbf{M}_{n_z}(\mathbb{R})), \\
f_1^n & \xrightarrow[n \to \infty]{} f_1, \quad \textnormal{in}\; L^\infty(\mathbb{R}_{\geq 0};\mathbb{R}^{n_X}), \\
f_2^n & \xrightarrow[n \to \infty]{} f_2, \quad \textnormal{in}\; L^\infty(\mathbb{R}_{\geq 0};L^2((0,1);\mathbb{R}^{n_z})).
\end{align}
\end{subequations}
Similar arguments as those in the proof of Theorem~\ref{thm:mild} ensure the existence of a unique (global) classical solution $(\tilde{X}^n,\tilde{z}^n) \in C^1(\mathbb{R}_{\geq 0};\altmathcal{X}) \times C^0(\mathbb{R}_{\geq 0}; \mathscr{D}(\mathscr{A}))$ to~\eqref{eq:obSGAIIn} for all ICs $(\tilde{X}_0^n,\tilde{z}_0^n) \triangleq (\tilde{X}^n(0),\tilde{z}^n(\cdot,0)) \in \mathscr{D}(\mathscr{A})$.

Lemma~\ref{lemma:convergence} below is propaedeutic to deriving the main result of this section.
\begin{lemma}\label{lemma:convergence}
Consider the mild solution $(\tilde{X},\tilde{z}) \in C^0([0,T];\altmathcal{X})$ of~\eqref{eq:obSGAII}, along with the classical solution $(\tilde{X}^n,\tilde{z}^n) \in C^1([0,T];\altmathcal{X})\cap C^0([0,T];\mathscr{D}(\mathscr{A}))$ of~\eqref{eq:obSGAIIn}. Then, for all and $(X,z) \in C^0(\mathbb{R}_{\geq 0};\altmathcal{X})\cap L^\infty(\mathbb{R}_{\geq 0};\altmathcal{X})$, $U \in C^0(\mathbb{R}_{\geq 0}; \mathbb{R}^{n_U})\cap L^\infty(\mathbb{R}_{\geq 0}; \mathbb{R}^{n_U})$ and $\hat{\Theta},\tilde{\Theta} \in C^1(\mathbb{R}_{\geq 0};\mathbf{Sym}_{n_z}(\mathbb{R}))\cap L^\infty(\mathbb{R}_{\geq 0};\mathbf{Sym}_{n_z}(\mathbb{R}))$, and ICs $(\tilde{X}_0,\tilde{z}_0)\in \altmathcal{X}$ and $(\tilde{X}_0^n,\tilde{z}_0^n)\in \mathscr{D}(\mathscr{A})$ satisfying
\begin{align}\label{eq:Xnnnnnn0}
(\tilde{X}_0^n, \tilde{z}_0^n) & \xrightarrow[n \to \infty]{} (\tilde{X}_0,\tilde{z}_0), \quad \textnormal{in}\; \altmathcal{X}, 
\end{align}
there exists $\omega_\rho \in \mathbb{R}_{> 0}$ such that
\begin{align}\label{eq:supCobv}
\sup_{t\in\mathbb{R}_{\geq 0}}\eu^{-\omega_\rho t}\norm{(\tilde{X}(t),\tilde{z}(\cdot,t))- (\tilde{X}^n(t),\tilde{z}^n(\cdot,t))}_{\altmathcal{X}} \xrightarrow[n \to \infty]{} 0.
\end{align}
\begin{proof}
See Appendix~\ref{app:cond}.
\end{proof}
\end{lemma}

\begin{theorem}\label{thmObse}
Suppose that $\Sigma \in C^0(\mathbb{R}^{n_z};\mathbf{M}_{n_z}(\mathbb{R}))$, $h_1, h_2 \in C^0(\mathbb{R}^{n_z};\mathbb{R}^{n_z})$ are locally Lipschitz continuous, and $U \in C^0([0,T];\mathbb{R}^{n_U})$, and consider the mild solutions $(X,z), (\tilde{X},\tilde{z}) \in C^0([0,T];\altmathcal{X})$ of~\eqref{eq:originalSystems} and~\eqref{eq:obSGAII}, respectively, together with the adaptive law~\eqref{eq:dottheta} and~\eqref{eq:RQ}. Then, if Assumptions~\ref{ass:fwd} and~\ref{ass:rightInv} hold, $(\tilde{X},\tilde{z}) \in L^\infty(\mathbb{R}_{\geq 0}; \altmathcal{X})$. Moreover, if $\varphi \in L^\infty(\mathbb{R}_{\geq 0};\mathbb{R}^{n_z})$ in~\eqref{eq:varphi} is PE, $\norm{(\tilde{X}(t), \tilde{z}(\cdot,t))}_{\altmathcal{X}} \to 0$ exponentially fast.

\begin{proof}
The strategy consists of proving the result for the regularized ODE-PDE system~\eqref{eq:obSGAIIn}, for whose classical solutions a suitable energy estimate may be derived using standard Lyapunov arguments. Resorting to~\eqref{eq:supCobv}, the obtained bound is then extended to mild solutions of ~\eqref{eq:obSGAII}.
To this end, the following Lyapunov function candidate is considered:
\begin{align}\label{eq:LyapReg}
\begin{split}
V\bigl(\tilde{X}^n(t),\tilde{z}^n(\cdot,t)\bigr)& \triangleq \dfrac{1}{2}\tilde{X}^{n\mathrm{T}}(t)P\tilde{X}^n(t) \\
& \quad +\dfrac{\gamma}{2} \int_0^1 \eu^{-\alpha\xi}\tilde{z}^{n\mathrm{T}}(\xi,t)\tilde{z}^n(\xi,t) \dif \xi,
\end{split}
\end{align}
for a positive definite matrix $\mathbf{Sym}_{n_X}(\mathbb{R}) \ni P \succ0$ and constants $\alpha, \gamma \in \mathbb{R}_{>0}$ to be appropriately specified.

Taking the derivative of~\eqref{eq:LyapReg} along the dynamics~\eqref{eq:obSGAIIn} yields
\begin{align}\label{eq:Vpppp}
\begin{split}
\dot{V}(t) &= \dfrac{1}{2}\tilde{X}^{n\mathrm{T}}(t)\Bigl(\bar{A}_1^{\mathrm{T}}P + P\bar{A}_1\Bigr)\tilde{X}^n(t) +\tilde{X}^{n\mathrm{T}}(t)Pf_1^n(t) \\
& \quad + \tilde{X}^{n\mathrm{T}}(t)P\bigl[G_1(\mathscr{K}_1\tilde{z}^n)(t) + F_1^n(t)(\mathscr{K}_2\tilde{z}^n)(t) \bigr] \\
& \quad - \gamma\int_0^1\eu^{-\alpha\xi}\tilde{z}^{n\mathrm{T}}(\xi,t)\Lambda\dpd{\tilde{z}^n(\xi,t)}{\xi} \dif \xi \\
& \quad +  \gamma\int_0^1\eu^{-\alpha\xi}\tilde{z}^{n\mathrm{T}}(\xi,t)F_2^n(t)\tilde{z}^n(\xi,t) \dif \xi \\
& \quad +\gamma\int_0^1\eu^{-\alpha\xi}\tilde{z}^{n\mathrm{T}}(\xi,t)f_2^n(\xi,t) \dif \xi, \quad t\in(0,T).
\end{split}
\end{align}
Assumption~\ref{ass:rightInv} implies that the matrix $P$ may be selected such that $\bar{A}_1^{\mathrm{T}}P + P\bar{A}_1 = -2pI_{n_X}$, with $p \in \mathbb{R}_{>0}$ aribitrarily large. Moreover, integrating by parts the fourth term in~\eqref{eq:Vpppp} and using the BC~\eqref{eq:originalSystemsBCObsErrLn} provides
\begin{align}\label{eq:Vpppp2}
\begin{split}
\dot{V}(t) &= -p\norm{\tilde{X}^n(t)}_2^2+\tilde{X}^{n\mathrm{T}}(t)Pf_1^n(t)  \\
& \quad + \tilde{X}^{n\mathrm{T}}(t)P\bigl[G_1(\mathscr{K}_1\tilde{z}^n)(t) + F_1^n(t)(\mathscr{K}_2\tilde{z}^n)(t) \bigr] \\
& \quad -\dfrac{\gamma}{2}\eu^{-\alpha}\tilde{z}^{n\mathrm{T}}(1,t)\Lambda\tilde{z}^n(1,t)\\
& \quad -\dfrac{\gamma\alpha}{2} \int_0^1\eu^{-\alpha\xi}\tilde{z}^{n\mathrm{T}}(\xi,t)\Lambda\tilde{z}^n(\xi,t) \dif \xi \\
& \quad +\gamma  \int_0^1\eu^{-\alpha\xi}\tilde{z}^{n\mathrm{T}}(\xi,t)F_2^n(t)\tilde{z}^n(\xi,t) \dif \xi \\
& \quad +\gamma\int_0^1\eu^{-\alpha\xi}\tilde{z}^{n\mathrm{T}}(\xi,t)f_2^n(\xi,t) \dif \xi, \quad t\in(0,T).
\end{split}
\end{align}
From~\eqref{eq;operatorsKs}, it may also be inferred that there exist $\eta_1,\eta_2 \in \mathbb{R}_{>0}$ such that
\begin{subequations}
\begin{align}
(\mathscr{K}_1\tilde{z}^n)(t)& \leq \eta_1\bigl(\norm{\tilde{z}^n(\cdot,t)}_{L^2((0,1);\mathbb{R}^{n_z})} + \norm{\tilde{z}^n(1,t)}_2\bigr), \\
(\mathscr{K}_2\tilde{z}^n)(t)& \leq \eta_2\norm{\tilde{z}^n(\cdot,t)}_{L^2((0,1);\mathbb{R}^{n_z})}.
\end{align}
\end{subequations}
Moreover, by virtue of~\eqref{eq:Fns} and~\eqref{eq:Xnnnnnnnn}, there exist $\eta_3,\eta_4 \in \mathbb{R}_{>0}$ such that
\begin{subequations}
\begin{align}
\norm{F_1(\cdot)}_\infty, \norm{F_1^n(\cdot)}_\infty & \leq \eta_3, \\
\norm{F_2(\cdot)}_\infty, \norm{F_2^n(\cdot)}_\infty & \leq \eta_4, \quad n \in \mathbb{N}_0.
\end{align}
\end{subequations}
Consequently, the first two cross terms appearing in~\eqref{eq:Vpppp2} may be bounded as
\begin{subequations}\label{eq;bound1}
\begin{align}
& \tilde{X}^{n\mathrm{T}}(t)Pf_1^n(t) \leq \dfrac{p}{4}\norm{\tilde{X}^n(t)}_2^2 + \dfrac{\norm{P}^2}{p}\norm{f_1^n(t)}_2^2, \\
\begin{split}
& \tilde{X}^{n\mathrm{T}}(t)P\bigl[G_1(\mathscr{K}_1\tilde{z}^n)(t) + F_1^n(t)(\mathscr{K}_2\tilde{z}^n)(t) \bigr]\leq \dfrac{p}{4}\norm{\tilde{X}^n(t)}_2^2 \\
& \quad  + \dfrac{2\norm{P}^2\eta_5^2}{p}\Bigl(\norm{\tilde{z}^n(\cdot,t)}_{L^2((0,1);\mathbb{R}^{n_z})}^2 + \norm{\tilde{z}^n(1,t)}_2^2 \Bigr),
\end{split}
\end{align}
\end{subequations}
with $\mathbb{R}_{> 0} \ni \eta_5 \triangleq \norm{G_1}\eta_1 + \eta_2\eta_3$. Similarly, the last two cross terms in~\eqref{eq:Vpppp2} may be bounded as
\begin{subequations}\label{eq;bound2}
\begin{align}
\begin{split}
& \int_0^1\eu^{-\alpha\xi}\tilde{z}^{n\mathrm{T}}(\xi,t)F_2^n(t)\tilde{z}^n(\xi,t) \dif \xi \\
& \quad \leq \eta_4\int_0^1\eu^{-\alpha\xi}\norm{\tilde{z}^n(\xi,t)}_2^2 \dif \xi, 
\end{split}\\
\begin{split}
& \int_0^1\eu^{-\alpha\xi}\tilde{z}^{n\mathrm{T}}(\xi,t)f_2^n(\xi,t) \dif \xi \\
& \quad \leq \dfrac{\alpha\lambda\ped{min}(\Lambda)}{4}\int_0^1\eu^{-\alpha\xi}\norm{\tilde{z}^n(\xi,t)}_2^2 \dif \xi  \\
& \qquad  + \dfrac{1}{\alpha\lambda\ped{min}(\Lambda)}\norm{f_2^n(\cdot,t)}_{L^2((0,1);\mathbb{R}^{n_z})}^2,
\end{split}
\end{align}
\end{subequations}
with $\mathbb{R}_{>0} \ni \lambda\ped{min}(\Lambda)$ denoting the smallest eigenvalue of $\Lambda$.
Combining~\eqref{eq;bound1},~\eqref{eq;bound2}, and~\eqref{eq:Vpppp2}, and selecting $\alpha \in \mathbb{R}_{>0}$ satisfying $\alpha > \frac{2\eta_4}{\lambda\ped{min}(\Lambda)}$ ensures the existence of $\eta_6 \in \mathbb{R}_{>0}$ such that
\begin{align}\label{eq:Vpppp3}
\begin{split}
\dot{V}(t) &= -\dfrac{p}{2}\norm{\tilde{X}^n(t)}_2^2+\dfrac{\norm{P}^2}{p}\norm{f_1^n(t)}_2^2 \\
& \quad -\Biggl(\dfrac{\gamma}{2}\lambda\ped{min}(\Lambda)\eu^{-\alpha}-\dfrac{2\norm{P}^2\eta_5^2}{p}\Biggr)\norm{\tilde{z}^n(1,t)}_2^2\\
& \quad -\Biggl(\dfrac{\gamma\eta_6}{2} - \dfrac{2\norm{P}^2\eta_5^2}{p}\Biggr)\norm{\tilde{z}^n(\cdot,t)}_{L^2((0,1);\mathbb{R}^{n_z})}^2 \\
& \quad +\dfrac{\gamma}{\alpha\lambda\ped{min}(\Lambda)}\norm{f_2^n(\cdot,t)}_{L^2((0,1);\mathbb{R}^{n_z})}^2, \quad t\in(0,T).
\end{split}
\end{align}
Hence, choosing $\gamma \in \mathbb{R}_{>0}$ sufficiently large implies the existence of constants $\bar{\omega},\eta \in \mathbb{R}_{>0}$ such that
\begin{align}
\dot{V}(t) \leq - \bar{\omega} V(t) + \eta \norm{f^n(\cdot,t)}_{\altmathcal{X}}^2, \quad t\in(0,T),
\end{align}
with $\mathbb{R}^{n_X+n_z}\ni f^n(\xi,t) \triangleq [f_1^{n\mathrm{T}}(t)\; f_2^{n\mathrm{T}}(\xi,t)]^{\mathrm{T}}$.
Thus, applying Grönwall-Bellman's inequality and observing that the Lyapunov function $V(X^n(t),\tilde{z}^n(\cdot,t))$ is equivalent to the squared norm $\norm{(\tilde{X}^n(t), \tilde{z}^n(\cdot,t))}_{\altmathcal{X}}^2$ on $\altmathcal{X}$ provides
\begin{align}\label{eqLJJJD0}
\begin{split}
& \norm{(\tilde{X}^n(t), \tilde{z}^n(\cdot,t))}_{\altmathcal{X}}  \leq \beta_1\eu^{-\sigma t}\norm{(\tilde{X}_0^n, \tilde{z}_0^n(\cdot))}_{\altmathcal{X}} \\
& \quad + \beta_2 \sqrt{\int_0^t\eu^{-\bar{\omega}(t-t^\prime)}\norm{f^n(\cdot,t^\prime)}_{\altmathcal{X}}^2 \dif t^\prime}, \quad t\in[0,T],
\end{split}
\end{align}
for some $\beta_1,\beta_2, \sigma\in \mathbb{R}_{>0}$. Finally, recalling~\eqref{eq:supCobv}, letting $n \to \infty$ in~\eqref{eqLJJJD0}, and invoking the lower-semicontinuity property of the norms together with the Dominated Convergence Theorem produces
\begin{align}\label{eqLJJJD}
\begin{split}
& \norm{(\tilde{X}(t), \tilde{z}(\cdot,t))}_{\altmathcal{X}}  \leq \beta_1\eu^{-\sigma t}\norm{(\tilde{X}_0, \tilde{z}_0(\cdot))}_{\altmathcal{X}} \\
& \quad + \beta_2 \sqrt{\int_0^t\eu^{-\bar{\omega}(t-t^\prime)}\norm{f(\cdot,t^\prime)}_{\altmathcal{X}}^2 \dif t^\prime}, \quad t\in[0,T],
\end{split}
\end{align}
where $\mathbb{R}^{n_X+n_z}\ni f(\xi,t) \triangleq [f_1^{\mathrm{T}}(t)\; f_2^{\mathrm{T}}(\xi,t)]^{\mathrm{T}}$. The conclusion follows immediately from the bound~\eqref{eqLJJJD}, in conjunction with the assumed behavior of $\norm{f(\cdot,t)}_{\altmathcal{X}}$, which is dictated by that of $z(\xi,t)$, $(\mathscr{K}_2z)(t)$, and $\tilde{\Theta}(t)$.
\end{proof}
\end{theorem}

\begin{theorem}\label{thm:Thm33}
Suppose that $\Sigma \in C^1(\mathbb{R}^{n_z};\mathbf{M}_{n_z}(\mathbb{R}))$, $h_1, h_2 \in C^1(\mathbb{R}^{n_z};\mathbb{R}^{n_z})$, and $U \in C^1([0,T];\mathbb{R}^{n_U})$, and consider the classical solutions $(X,z), (\tilde{X},\tilde{z}) \in C^1([0,T];\altmathcal{X})\cap C^0([0,T];\mathscr{D}(\mathscr{A}))$ of~\eqref{eq:originalSystems} and~\eqref{eq:obSGAII}, respectively, together with the adaptive law~\eqref{eq:dottheta} and~\eqref{eq:RQ}. Then, if Assumptions~\ref{ass:fwd} and~\ref{ass:rightInv} hold, $(\tilde{X},\tilde{z}) \in L^\infty(\mathbb{R}_{\geq 0}; \altmathcal{Y})$. Moreover, if $\varphi \in L^\infty(\mathbb{R}_{\geq 0};\mathbb{R}^{n_z})$ in~\eqref{eq:varphi} is PE, $\norm{(\tilde{X}(t), \tilde{z}(\cdot,t))}_{\altmathcal{Y}} \to 0$ exponentially fast.
\begin{proof}
If $\Sigma \in C^1(\mathbb{R}^{n_z};\mathbf{M}_{n_z}(\mathbb{R}))$, $h_1 \in C^1(\mathbb{R}^{n_z};\mathbb{R}^{n_X})$, $h_2 \in C^1(\mathbb{R}^{n_z};\mathbb{R}^{n_z})$, and $U \in C^1([0,T];\mathbb{R}^{n_U})$, the observer error dynamics~\eqref{eq:obSGAII} indeed admits a unique (global) classical solution $(\tilde{X},\tilde{z}) \in C^1(\mathbb{R}_{\geq 0};\altmathcal{X})\cap C^0(\mathbb{R}_{\geq 0};\mathscr{D}(\mathscr{A}))$. Therefore, the bound~\eqref{eqLJJJD} may be derived by considering directly the Lyapunov function $V(\tilde{X}(t),\tilde{z}(\cdot,t))$ defined according to~\eqref{eq:LyapReg}, implying that $(\tilde{X},\tilde{z}) \in L^\infty(\mathbb{R}_{\geq 0};\altmathcal{X})$, and $\norm{(\tilde{X}(t),\tilde{z}(\cdot,t))}_{\altmathcal{X}}\to 0$ exponentially fast whenever $\varphi \in L^\infty(\mathbb{R}_{\geq 0};\mathbb{R}^{n_z})$ in~\eqref{eq:varphi} is PE. Additionally, defining $\mathbb{R}^{n_z}\ni \tilde{\zeta}(\xi,t) \triangleq \pd{\tilde{z}(\xi,t)}{\xi}$ and taking the partial derivative of~\eqref{eq:originalSystemsPDEObsErrL} yields
\begin{subequations}\label{eq:dyn2}
\begin{align}
\begin{split}
& \dpd{\tilde{\zeta}(\xi,t)}{t} + \Lambda\dpd{\tilde{\zeta}(\xi,t)}{\xi} = F_2(t)\tilde{\zeta}(\xi,t) + f_3(\xi,t), \\
& \qquad \qquad \qquad \qquad \qquad \qquad (\xi,t)\in(0,1)\times(0,T), 
\end{split}\\
& \tilde{\zeta}(0,t) = 0, \quad t \in(0,T),\label{eq:originalSystemsBCObsErrLn2}
\end{align}
\end{subequations}
with
\begin{align}
f_3(\xi,t) & \triangleq  \dpd{f_2(\xi,t)}{\xi} = \tilde{\Theta}(t)\Bigl(v\bigl(X(t),U(t)\bigr)\Bigr) \dpd{z(\xi,t)}{\xi}.
\end{align}
In conjunction with Assumption~\ref{ass:fwd}, $(X,z)\in C^1(\mathbb{R}_{\geq 0};\altmathcal{X})\cap C^0(\mathbb{R}_{\geq 0};\mathscr{D}(\mathscr{A}))$ implies that $f_3 \in C^0(\mathbb{R}_{\geq 0};\altmathcal{X})\cap L^\infty(\mathbb{R}_{\geq 0};\altmathcal{X})$. 

Hence, for classical solutions of~\eqref{eq:dyn2}, using standard Lyapunov arguments, it is possible to derive an energy estimate of the form
\begin{align}\label{eqLJJJD2}
\begin{split}
& \norm{\tilde{\zeta}(\cdot,t)}_{L^2((0,1);\mathbb{R}^{n_z})}  \leq \beta_3\eu^{-\sigma^\prime t} \norm{\tilde{\zeta}_0(\cdot)}_{L^2((0,1);\mathbb{R}^{n_z})} \\
& \quad + \beta_4 \sqrt{\int_0^t\eu^{-\bar{\omega}^\prime(t-t^\prime)}\norm{f_3(\cdot,t^\prime)}_{L^2((0,1);\mathbb{R}^{n_z})}^2 \dif t^\prime}, \quad t\in[0,T],
\end{split}
\end{align}
for some $\beta_3, \beta_4, \bar{\omega}^\prime, \sigma^\prime \in \mathbb{R}_{>0}$. The above bound~\eqref{eqLJJJD2} may be extended to mild solutions of~\eqref{eq:dyn2} following a similar rationale as that adopted in the proof of Lemma~\ref{lemma:convergence}. The conclusion is an immediate consequence of the estimate~\eqref{eqLJJJD2} and the definition of $\tilde{\zeta}(\xi,t)$.
\end{proof}
\end{theorem}
Theorem~\ref{thm:Thm33} concludes the technical part of the paper. Before moving to Section~\ref{sect:example}, it is worth remarking that Theorems~\ref{thmObse} and~\ref{thm:Thm33} require PE conditions for the signal $\varphi(t)$ to ensure the convergence of the observer error dynamics in the corresponding norms. In practice, PE conditions need to be checked numerically; preliminary insights on how to design the control input $U(t)$ to ensure PE of $\varphi(t)$ may also be gained by employing reduced-order representations of the ODE-PDE system~\eqref{eq:originalSystems}, see, e.g., \textcite{Spert}.

\section{An example from vehicle dynamics}\label{sect:example}
The proposed adaptive observer is validated considering a relevant example borrowed from vehicle dynamics. 

\subsection{Lateral vehicle dynamics}\label{sect:exampleEqs} 
The performance of the adaptive observer designed in Section~\ref{sect:observer} is tested on a semilinear single-track model governing the lateral dynamics of a road vehicle driving at a constant cruising speed. A schematic of the model is depicted in Figure~\ref{figureForcePostdoc}.
\begin{figure}
\centering
\includegraphics[width=0.7\linewidth]{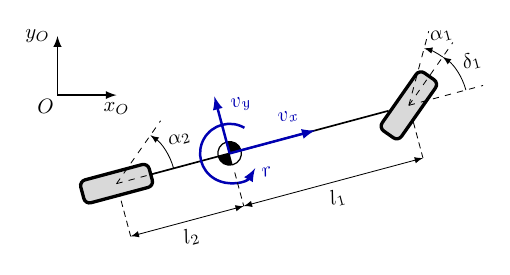} 
\caption{Single-track vehicle model.}
\label{figureForcePostdoc}
\end{figure}
In particular, for sufficiently small steering inputs, the linear ODE describing the rigid vehicle dynamics may be derived as \parencite{Guiggiani}
\begin{subequations}\label{eq:rigid}
\begin{align}
\dot{v}_y(t) & = -\dfrac{1}{m}\bigl(F_{y1}(t) + F_{y2}(t)\bigr) -v_xr(t), \\
\dot{r}(t) & = -\dfrac{1}{I_z}\bigl( l_1F_{y1}(t)-l_2F_{y2}(t)\bigr), && t\in (0,T),
\end{align}
\end{subequations}
where the lumped states $v_y(t)$, $r(t) \in \mathbb{R}$ are the vehicle's lateral velocity and yaw rate, $v_x\in \mathbb{R}_{>0}$ is its constant longitudinal speed, $m\in \mathbb{R}_{>0}$ and $I_z\in \mathbb{R}_{>0}$ denote respectively the vehicle mass and moment of inertia of the center of gravity around the vertical axis, and $l_1$, $l_2 \in \mathbb{R}_{>0}$ are the front and rear axle lengths. Finally, adopting a distributed Dahl model for dry friction \parencite{DistrLuGre}, the tire forces $F_{y1}(t)$, $F_{y2}(t) \in \mathbb{R}$ may be calculated as
\begin{align}\label{eq:Fi}
\begin{split}
F_{yi}(t) & = L_i\int_0^1 p_i(\xi)\sigma_{i}z_i(\xi,t), \quad i \in \{1,2\},
\end{split} 
\end{align}
where the distributed state $z_i(\xi,t) \in \mathbb{R}$, $i = \{1,2\}$, represents the deflection of a bristle schematizing a tire rubber particle inside the contact patch, $L_i \in \mathbb{R}_{>0}$ denotes the contact patch length, $p_i \in C^1([0,1];\mathbb{R}_{\geq 0})$ is the vertical pressure distribution, and $\sigma_{i} \in \mathbb{R}_{>0}$ is the normalized micro-stiffness coefficient \parencite{DistrLuGre}. The bristle dynamics is governed by the following PDE:
\begin{subequations}\label{Eq:PDEz}
\begin{align}
\begin{split}
& \dpd{z_i(\xi,t)}{t} + \dfrac{v_x}{L_i}\dpd{z_i(\xi,t)}{\xi} =-\dfrac{\theta_i\sigma_{i}\abs{v_i\bigl(v_y(t),r(t), \delta_i(t)\bigr)}}{\mu_i\Bigl(v_i\bigl(v_y(t),r(t), \delta_i(t)\bigr)\Bigr)}z_i(\xi,t) \\
& \qquad \qquad \qquad \qquad \qquad \qquad + 2 v_i\bigl(v_y(t),r(t), \delta_i(t)\bigr), \\
& \qquad \qquad \qquad \qquad \qquad \qquad (\xi,t) \in (0,1)\times (0,T),
\end{split} \\
& z_i(0,t) = 0, \quad t \in (0,T) \label{eq:BCz}
\end{align}
\end{subequations}
where $\theta_i \in \mathbb{R}_{>0}$ is the uncertain friction parameter, $\mu_i \in C^1(\mathbb{R};[\mu\ped{min},\infty))$, $i \in \{1,2\}$, with $\mu\ped{min}\in \mathbb{R}_{>0}$, is an expression for the nominal friction coefficient, and
\begin{subequations}\label{eq:slipAngles}
\begin{align}
\begin{split}
v_1\bigl(v_y(t),r(t), \delta_1(t)\bigr)  &= v_y(t) + l_1r(t)-v_x\delta_1(t), 
\end{split} \\
\begin{split}
v_2\bigl(v_y(t),r(t), \delta_2(t)\bigr) & = v_y(t)-l_2r(t)-v_x\delta_2(t),
\end{split}
\end{align}
\end{subequations}
being $\delta_1(t), \delta_2(t) \in \mathbb{R}$ the steering inputs at the front and rear axles, respectively.
Hence, defining $\mathbb{R}^2 \ni X(t) \triangleq [v_y(t)\; r(t)]^{\mathrm{T}}$, $\mathbb{R}^2\ni z(\xi,t) \triangleq [z_1(\xi,t) \; z_2(\xi,t)]^{\mathrm{T}}$, $\mathbb{R}^2 \ni U(t) \triangleq [\delta_1(t)\; \delta_2(t)]^{\mathrm{T}}$, and $\mathbb{R}^2 \ni v(X(t),U(t)) = [v_1(X(t),U(t))\; v_2(X(t),U(t))]^{\mathrm{T}} \triangleq [v_1(v_y(t),r(t), \delta_1(t)) \; v_2(v_y(t),r(t), \delta_2(t))]^{\mathrm{T}}$, the resulting system may be recast in the form~\eqref{eq:originalSystems}, with
\begin{align}\label{eq:MatricesOOO}
A_1 & \triangleq \begin{bmatrix} 0 & -v_x \\ 0 &  0 \end{bmatrix}, \quad A_2 \triangleq \begin{bmatrix}1 & l_1 \\ 1 & -l_2 \end{bmatrix},   \nonumber\\
G_1 & \triangleq -\begin{bmatrix}\dfrac{1}{m} & \dfrac{1}{m} \\ \dfrac{l_1}{I_z} & -\dfrac{l_2}{I_z} \end{bmatrix},\quad G_2  \triangleq -v_xI_2,  \nonumber\\
K_1(\xi) & \triangleq \begin{bmatrix} L_1\sigma_{1}p_1(\xi) & 0 \\ 0 & L_2\sigma_{2}p_2(\xi)\end{bmatrix}, \quad \Lambda \triangleq \begin{bmatrix} \dfrac{v_x}{L_1} & 0 \\ 0 & \dfrac{v_x}{L_2}\end{bmatrix}, \nonumber \\
 \Sigma(v) & \triangleq \begin{bmatrix} -\dfrac{\sigma_{1}\abs{v_1}}{\mu_1(v_1)} & 0 \\0 & -\dfrac{\sigma_{2}\abs{v_2}}{\mu_2(v_2)} \end{bmatrix}, \quad h_2(v)  \triangleq 2v, 
\end{align}
$\mathbf{Sym}_2(\mathbb{R}) \ni \Theta = \diag\{\theta_1,\theta_2\}$, $K_2=K_3=0$, and $h_1(v) = 0$. Finally, the vertical pressure distribution inside the tires' contact patches, appearing in~\eqref{eq:Fi} and~\eqref{eq:MatricesOOO}, may be modeled using exponentially decreasing functions of the type \parencite{DistrLuGre}
\begin{align}\label{eq:PressureDistr}
p_i(\xi) = p_{0,i}\exp(-a_i\xi), \quad \xi \in [0,1],
\end{align}
with $p_{0,i}$, $a_i \in \mathbb{R}_{>0}$, $i \in \{1,2\}$.

The single-track model formulated according to~\eqref{eq:originalSystems},~\eqref{eq:MatricesOOO}, and~\eqref{eq:PressureDistr} may be used to estimate the friction coefficients at the front and rear axles, $\theta_1$ and $\theta_2$, respectively, jointly with the vehicle states. Reliable information about tire-road friction is crucial not only for enhancing the internal performance of ADAS, but also to improve fleet and traffic flow coordination via vehicle-to-vehicle communication, and to support predictive operations such as road maintenance and salting \parencite{Salt,Albin}.
Starting with the measurements $Y_1(t) = \od{z(0,t)}{t}$ and $Y_2(t) = \od[2]{z(0,t)}{t}$, which may be acquired using smart tire sensors mounted on both the front and rear axles, it is straightforward to verify that Assumption~\ref{ass:rightInv} holds for any combination of model parameters. In this context, it is perhaps worth mentioning that the yaw rate $r(t)$ can be acquired with standard instrumentation equipped on any passenger vehicle, thus rendering one of the two boundary measurements redundant in most cases, especially if $\theta_1 = \theta_2$ is assumed.

\begin{remark}[Verification of PE via time-scale separation]\label{rem:pe-sp}
For the considered single-track model, the PE assumption on the filtered regressor
$\varphi(t)\in\mathbb{R}^{n_z}$ in Theorem~\ref{thm:3.1} can be checked
by exploiting the time-scale separation of the coupled ODE-PDE dynamics.
Indeed, introducing the time-scale parameter $\mathbb{R}_{>0} \ni \tau_z \triangleq \frac{\bar{L}}{v_x}$, where $ \bar{L}\triangleq \frac{L_1+L_2}{2}$ is a characteristic length of the problem, the distributed tire dynamics~\eqref{eq:originalSystemsPDE} can be written as $
\tau_z \pd{z(\xi,t)}{t} + \Lambda\pd{z(\xi,t)}{\xi}
= \Theta\Sigma(v(X(t),U(t)))z(\xi,t) + h_2(v(X(t),U(t))$, revealing a fast PDE subsystem coupled to a slow ODE subsystem. For sufficiently
small $\tau_z$ (i.e.,\ $v_x\gg \bar L$), the fast dynamics converge
exponentially to their quasi-steady profile~\parencite[Lemma 2.1]{Spert}, yielding the reduced-order model $
\dot{\bar X}(t)=A_1\bar X(t)+G_1\Phi(v(\bar X(t),\bar U(t)))$ \parencite[Section 3]{Spert}.
Under classic controllability conditions for the reduced-order model, standard manoeuvres (e.g., lane changes,
sine-sweep steering) can be designed to ensure that the rigid relative velocity $v(X(t),U(t))$ is persistently
excited, so that the filtered regressor satisfies the PE condition $
\int_t^{t+T}\varphi(t^\prime)\varphi^{\mathrm{T}}(t^\prime)\dif t^\prime \geq \mu I_{n_z}
$ for some $T,\mu \in \mathbb{R}_{>0}$ and all $t \in \mathbb{R}_{\geq 0}$.
By Theorem~3.1 in~\textcite{Spert}, for sufficiently small $\tau_z$
the full ODE-PDE system remains close to the reduced
dynamics, thereby preserving controllability and the PE property.
\end{remark}

\subsection{Simulation results}\label{sect:sim}

The numerical values for the model parameters of the example discussed below are $v_x = 20$ $\textnormal{m}\,\textnormal{s}^{-1}$, $m = 1300$ kg, $I_z = 2000$ $\textnormal{kg}\,\textnormal{m}^{2}$, $l_1 = 1$, $l_2 = 1.6$ m, $L_1 = 0.11$, $L_2 = 0.09$ m, $\sigma_1 = 165$, $\sigma_2 = 415$ $\textnormal{m}^{-1}$, $\mu_1(\cdot) = \mu_2(\cdot) = 1$, $p_{0,1} = 3.75\cdot 10^4$, $p_{0,2} = 2.86\cdot 10^4$ $\textnormal{N}\,\textnormal{m}^{-1}$, and $a_1 = a_2 = 0.1$.
With the given combination of parameters, the single-track model is understeer \parencite{Guiggiani}, and thus intrinsically stable for sufficiently low longitudinal speeds $v_x$, and small steering inputs $\delta_1(t)$ and $\delta_2(t)$. The following results refer to numerical simulations conducted in MATLAB/Simulink\textsuperscript{\textregistered} environment. The semilinear PDE subsystem was discretized in space using finite differences with a discretization step of $0.02$, whereas the fixed time step was specified as $10^{-6}$ s. The ICs for the actual system were set to $X_0 = [3\; -0.4]^{\mathrm{T}}$, and $z_0(\xi) = [0.3\; 0.3]^{\mathrm{T}}$ (corresponding to $\norm{z_0(\cdot)}_{L^2((0,1);\mathbb{R}^2)} = 0.42$), whereas those for the observer to $\hat{X}_0 = [0\; 0]^{\mathrm{T}}$, $\hat{z}_0(\xi) = [0\; 0]^{\mathrm{T}}$, and $\mathbb{R}^2 \ni \hat{\theta}_0 \triangleq \hat{\theta}(0) = [0\; 0]^{\mathrm{T}}$. The observer gain $L_1 \in \mathbf{M}_2(\mathbb{R})$ in~\eqref{eq:observer} was specified as $L_1 = -(A_1 + qI_2)A_2^{-1}$, with $q = 50$; the gains in~\eqref{eq:filters},~\eqref{eq:dottheta}, and~\eqref{eq:RQ} were instead chosen as $\varrho = 20$ and $\psi = 50$, and $\Gamma = 5000I_2$. 
Finally, a sinusoidal steering input, which is common in vehicle dynamics \parencite{Shao,Shao2}, was designed with $\delta_1(t) = 0.05\sin(2t) + 0.01\sin(4t)$ and $\delta_2(t) = 0$, and persistence of excitation was tested numerically.

Figure~\ref{figure:parameters} illustrates the convergence of the parameter estimate $\hat{\theta}(t)$ to the true value, corresponding to $\theta = [1.2\; 0.8]^{\mathrm{T}}$ for $t< 5$ s, and to $\theta = [0.6\; 0.4]^{\mathrm{T}}$ for $t\geq 5$ s. The abrupt change at $t = 5$ s simulates a sudden discontinuity in friction. From Figure~\ref{figure:parameters}, it is evident that a good estimate of $\theta$ is already achieved before $t = 2$ s during the first phase; during the second phase, the convergence is reached around $t = 8$ s.
\begin{figure}
\centering
\includegraphics[width=0.8\linewidth]{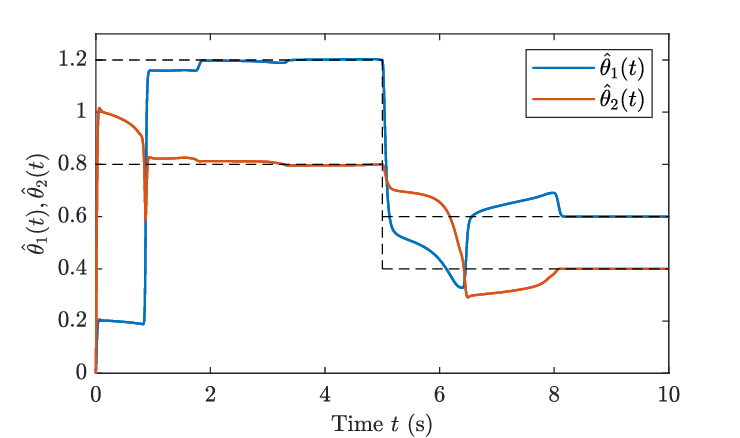} 
\caption{Convergence of the parameter estimate $\hat{\theta}(t)$ to the true value $\theta$.}
\label{figure:parameters}
\end{figure}
Due to the presence of the absolute value $\abs{\cdot}$ in~\eqref{Eq:PDEz} and in the matrix $\Sigma(v)$ in~\eqref{eq:MatricesOOO}, the source term is intrinsically non-smooth. Hence, the convergence of the observer error states should be investigated in the norm $\norm{(\tilde{X}(t),\tilde{z}(\cdot,t))}_{\altmathcal{X}}$.
In particular, the theoretical predictions of Theorem~\ref{thmObse} were also corroborated in simulation, as demonstrated graphically in Figure~\ref{figure:error}, where the norms $\norm{(X(t),z(\cdot,t))}_{\altmathcal{X}}$, $\norm{(\hat{X}(t),\hat{z}(\cdot,t))}_{\altmathcal{X}}$, and $\norm{(\tilde{X}(t),\tilde{z}(\cdot,t))}_{\altmathcal{X}}$ are plotted together ($t \leq 2$ s). A satisfactory estimate of the lumped and distributed states is achieved around $t = 0.25$ s, as indicated by the rapid initial decrease of $\norm{(\tilde{X}(t),\tilde{z}(\cdot,t))}_{\altmathcal{X}}$ to a value near zero. Subsequently, the error remains extremely small, and only exhibits a minor transient after the abrupt change in friction occurring at $t = 5$ s. 
\begin{figure}
\centering
\includegraphics[width=0.8\linewidth]{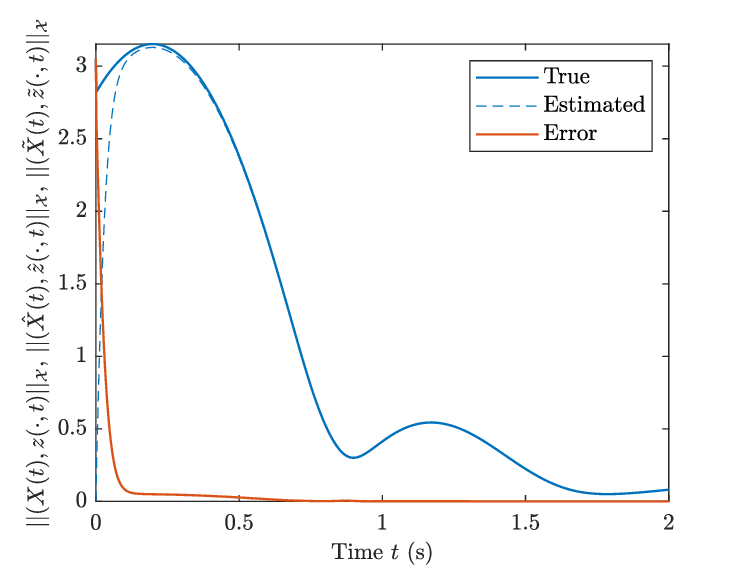} 
\caption{True states (solid tick blue line), observer estimates (dotted blue line), and observer error (solid orange line).}
\label{figure:error}
\end{figure}
Finally, it should be emphasized that the proposed observer for the PDE variables relies on exact cancellation via transport dynamics. To evaluate the robustness of the approach, additional simulations were performed by replacing the true transport velocity in~\eqref{eq:SystemsPDEObs} with an estimate obtained using a deliberately mismatched longitudinal speed, $\hat{v}_x = 0.8v_x$. The results, omitted here due to space restrictions, confirmed that the observer maintains satisfactory performance in the presence of significant parameter uncertainty.

\section{Conclusions}\label{sect:conclusion}
This paper focused on the synthesis of an adaptive observer for a class of semilinear ODE-PDE systems representing mechanical systems with distributed friction and rolling contact dynamics. Specifically, a boundary-sensing-based observer was designed that achieves joint estimation of both the lumped and distributed states, as well as uncertain friction parameters. Under persistence of excitation conditions, theoretical guarantees of convergence were provided for both mild and classical solutions, and validated through simulations on a vehicle dynamics benchmark. Future work will focus on experimental validation and extension to bidirectional PDE systems and multi-contact configurations.

\subsubsection*{Acknowledgments}
This research was financially supported by the project FASTEST (Reg. no. 2023-06511), funded by the Swedish Research Council.

{\setstretch{1} 
\printbibliography

@inproceedings{TsiotrasConf,
  author = {Canudas de Wit, C. and Tsiotras, P.},
  title = {Dynamic Tire Friction Models for Vehicle Traction Control},
  booktitle = {Proceedings of the IEEE Conference on Decision and Control},
  address = {Phoenix, AZ, USA},
  volume = {4},
  pages = {3746--3751},
  year = {1999}
}

@article{Tsiotras1,
  author = {Canudas-de-Wit, C. and Tsiotras, P. and Velenis, E. and others},
  title = {Dynamic Friction Models for Road/Tire Longitudinal Interaction},
  journal = {Vehicle System Dynamics},
  volume = {39},
  number = {3},
  pages = {189--226},
  year = {2003}
}

@article{Tsiotras2,
  author = {Velenis, E. and Tsiotras, P. and Canudas-de-Wit, C. and Sorine, M.},
  title = {Dynamic Tyre Friction Models for Combined Longitudinal and Lateral Vehicle Motion},
  journal = {Vehicle System Dynamics},
  volume = {43},
  number = {1},
  pages = {3--29},
  year = {2005}
}

@article{Deur1,
  author = {Deur, J. and Asgari, J. and Hrovat, D.},
  title = {A 3D Brush-type Dynamic Tire Friction Model},
  journal = {Vehicle System Dynamics},
  volume = {42},
  number = {3},
  pages = {133--173},
  year = {2004}
}

@article{Deur2,
  author = {Deur, J. and Ivanovi{\'c}, V. and Troulis, M. and others},
  title = {Extensions of the LuGre Tyre Friction Model Related to Variable Slip Speed Along the Contact Patch Length},
  journal = {Vehicle System Dynamics},
  volume = {43},
  number = {sup1},
  pages = {508--524},
  year = {2005}
}

@article{Mio,
  author = {Romano, L. and Bruzelius, F. and Jacobson, B.},
  title = {An Extended LuGre-brush Tyre Model for Large Camber Angles and Turning Speeds},
  journal = {Vehicle System Dynamics},
  volume = {61},
  number = {6},
  pages = {1674--1706},
  year = {2022}
}

@article{Frendo1,
  author = {Frendo, F. and Bucchi, F.},
  title = {Brush Model for the Analysis of Flat Belt Transmissions in Steady-State Conditions},
  journal = {Mechanism and Machine Theory},
  volume = {143},
  pages = {103698},
  year = {2020}
}

@article{Frendo2,
  author = {Frendo, F. and Bucchi, F.},
  title = {Enhanced Brush Model for the Mechanics of Power Transmission in Flat Belt Drives Under Steady-State Conditions: Effect of Belt Elasticity},
  journal = {Mechanism and Machine Theory},
  volume = {153},
  pages = {104003},
  year = {2020}
}

@article{2D,
  author = {Waltersson, G. A. and Karayiannidis, Y.},
  title = {Planar Friction Modeling With LuGre Dynamics and Limit Surfaces},
  journal = {IEEE Trans. Robot.},
  volume = {40},
  pages = {3166--3180},
  year = {2024}
}

@article{DistrLuGre,
  author = {Romano, L. and Aamo, O. M. and {\AA}slund, J. and Frisk, E.},
  title = {Stability and Dissipativity of the Distributed LuGre Friction Model},
  journal = {IEEE Control Systems Letters},
  year = {2025},
}

@article{FrBD,
  author = {Romano, L. and Aamo, O. M. and {\AA}slund, J. and Frisk, E.},
  title = {First-order Friction Models with Bristle Dynamics: Lumped and Distributed Formulations},
  journal = {IEEE Trans. Contr. Syst. Tech.},
  year = {2026},
}

@article{Savaresi,
  author = {Tanelli, M. and Astolfi, A. and Savaresi, S. M.},
  title = {Robust Nonlinear Output Feedback Control for Brake-by-Wire Control Systems},
  journal = {Automatica},
  volume = {44},
  number = {4},
  pages = {1078--1087},
  year = {2008}
}

@inproceedings{Shao,
  author = {Shao, L. and Jin, C. and Lex, C. and Eichberger, A.},
  title = {Nonlinear Adaptive Observer for Side Slip Angle and Road Friction Estimation},
  booktitle = {Proceedings of the IEEE Conference on Decision and Control},
  address = {Las Vegas, NV, USA},
  pages = {6258--6265},
  year = {2016}
}

@article{Shao2,
  author = {Shao, L. and Jin, C. and Lex, C. and others},
  title = {Robust Road Friction Estimation During Vehicle Steering},
  journal = {Vehicle System Dynamics},
  volume = {57},
  number = {4},
  pages = {493--519},
  year = {2019}
}

@book{Smysh,
  author = {Smyshlyaev, A. and Krstic, M.},
  title = {Adaptive Control of Parabolic PDEs},
  publisher = {Princeton University Press},
  year = {2010}
}

@article{Krstic1,
  author = {Bernard, P. and Krstic, M.},
  title = {Adaptive Output-Feedback Stabilization of Non-local Hyperbolic PDEs},
  journal = {Automatica},
  volume = {50},
  number = {10},
  pages = {2692--2699},
  year = {2014}
}

@book{OleBook,
  author = {Anfinsen, H. and Aamo, O. M.},
  title = {Adaptive Control of Hyperbolic PDEs},
  publisher = {Springer},
  address = {Cham},
  year = {2019}
}

@article{Ole0,
  author = {Aamo, O. M.},
  title = {Disturbance Rejection in $2\times2$ Linear Hyperbolic Systems},
  journal = {IEEE Trans. Autom. Contr.},
  volume = {58},
  number = {5},
  pages = {1095--1106},
  year = {2013}
}

@article{Ghousein1,
  author = {Ghousein, M. and Witrant, E.},
  title = {Adaptive Observer Design for Uncertain Hyperbolic PDEs Coupled with Uncertain LTV ODEs; Application to Refrigeration Systems},
  journal = {Automatica},
  volume = {154},
  pages = {111233},
  year = {2023}
}

@article{Ben,
  author = {Benabdelhadi, A. and Giri, F. and Ahmed-Ali, T. and others},
  title = {Adaptive Observer Design for Wave PDEs with Nonlinear Dynamics and Parameter Uncertainty},
  journal = {Automatica},
  volume = {123},
  pages = {109331},
  year = {2021}
}

@article{Aamo,
  author = {Anfinsen, H. and Aamo, O. M.},
  title = {Leak Detection, Size Estimation and Localization in Branched Pipe Flows},
  journal = {Automatica},
  volume = {140},
  pages = {110109},
  year = {2022}
}

@inproceedings{Aamo3,
  author = {Anfinsen, H. and Diagne, M. and Aamo, O. M. and others},
  title = {An Adaptive Observer Design for $n+1$ Coupled Linear Hyperbolic PDEs Based on Swapping},
  booktitle = {IEEE Trans. Autom. Contr.},
  volume = {61},
  number = {12},
  pages = {3979--3990},
  year = {2016}
}

@inproceedings{Aamo2,
  author = {Anfinsen, H. and Aamo, O. M.},
  title = {State Estimation in Hyperbolic PDEs Coupled with an Uncertain LTI System},
  booktitle = {Proceedings of the American Control Conference},
  address = {Seattle, WA, USA},
  pages = {3821--3827},
  year = {2017}
}

@inproceedings{Ole1,
  author = {Anfinsen, H. and Aamo, O. M.},
  title = {Estimation of Parameters in a Class of Hyperbolic Systems with Uncertain Transport Speeds},
  booktitle = {Mediterranean Conference on Control and Automation},
  address = {Valletta, Malta},
  year = {2017}
}

@inproceedings{Ole2,
  author = {Anfinsen, H. and Holta, H. and Aamo, O. M.},
  title = {Adaptive Control of a Linear Hyperbolic PDE with Uncertain Transport Speed and a Spatially Varying Coefficient},
  booktitle = {Mediterranean Conference on Control and Automation},
  address = {Saint-Rapha{\"e}l, France},
  pages = {945--951},
  year = {2020}
}

@book{Coron,
  author = {Bastin, G. and Coron, J.-M.},
  title = {Stability and Boundary Stabilization of 1-D Hyperbolic Systems},
  publisher = {Birkh{\"a}user},
  address = {Cham},
  year = {2016}
}

@book{Kato,
  author = {Kato, T.},
  title = {Perturbation Theory for Linear Operators},
  publisher = {Springer},
  address = {Berlin},
  year = {1995}
}

@article{SemilinearMio,
  author = {Romano, L. and Aamo, O. M. and {\AA}slund, J. and Frisk, E.},
  title = {Semilinear Single-track Vehicle Models with Distributed Tyre Friction Dynamics},
  journal = {Nonlinear Dynamics},
  volume = {114},
  pages = {138},
  year = {2026}
}

@article{CanudasFr2,
  author = {Canudas-de-Wit, C. and Lischinsky, P.},
  title = {Adaptive Friction Compensation with Dynamic Friction Model},
  journal = {IFAC Proceedings Volumes},
  volume = {29},
  number = {1},
  pages = {2078--2083},
  year = {1996}
}

@inproceedings{CanudasFr1,
  author = {Canudas-de-Wit, C. and Horowitz, R.},
  title = {Observers for Tire/Road Contact Friction Using Only Wheel Angular Velocity Information},
  booktitle = {Proceedings of the IEEE Conference on Decision and Control},
  address = {Phoenix, AZ, USA},
  pages = {3932--3937},
  year = {1999}
}

@book{Ioannou,
  author = {Ioannou, P. and Sun, J.},
  title = {Robust Adaptive Control},
  publisher = {Prentice Hall},
  address = {Upper Saddle River, NJ},
  year = {1995}
}

@article{Spert,
  author = {Romano, L. and Aamo, O. M. and {\AA}slund, J. and Frisk, E.},
  title = {Stability and Stabilization of Semilinear Single-track Vehicle Models with Distributed Tire Friction Dynamics via Singular Perturbation Analysis},
  journal = {Automatica},
  year = {2025},
  note = {submitted}
}

@book{Guiggiani,
  author = {Guiggiani, M.},
  title = {The Science of Vehicle Dynamics},
  publisher = {Springer},
  address = {Cham},
  year = {2023}
}

@book{Salt,
  author = {Wielitzka, M. and Dagen, M. and Ortmaier, T.},
  title = {State and Maximum Friction Coefficient Estimation in Vehicle Dynamics Using UKF},
  booktitle = {Proceedings of the American Control Conference},
  year = {2017}
}

@article{Albin,
  author = {Albinsson, A. and Bruzelius, F. and Jacobson, B. and Fredriksson, J.},
  title = {Design of Tyre Force Excitation for Tyre--Road Friction Estimation},
  journal = {Vehicle System Dynamics},
  volume = {55},
  number = {2},
  pages = {208--230},
  year = {2016}
}

@book{Pazy,
  author = {Pazy, A.},
  title = {Semigroups of Linear Operators and Applications to Partial Differential Equations},
  publisher = {Springer},
  address = {New York},
  year = {1983}
}

@article{Traffic,
  author = {Wu, J. and Zhan, J. and Zhang, L.},
  title = {Adaptive Boundary Observers for Hyperbolic PDEs With Application to Traffic Flow Estimation},
  journal = {IEEE Trans. Auto. Contr.},
  volume = {69},
  number = {1},
  pages = {651--658},
  year = {2024}
}

@article{Traffic2,
  author = {Zhang, L. and Wu, J. and Zhan, J.},
  title = {Adaptive observer design for coupled ODE–hyperbolic PDE systems with application to traffic flow estimation},
  journal = {Automatica},
  volume = {167},
  number = {},
  pages = {},
  year = {2024}
}

@article{Hasan,
  author = {Hasan, A. and Aamo, O. M. and Krstic, M.},
  title = {Boundary observer design for hyperbolic PDE–ODE cascade systems},
  journal = {Automatica},
  volume = {67},
  number = {},
  pages = {75–86},
  year = {2016}
}
}

\appendix


\section{Proof of Lemma~\ref{lemma:convergence}}\label{app:cond}
The proof of Lemma~\ref{lemma:convergence} is given below.

\begin{proof}[Proof of Lemma~\ref{lemma:convergence}]
Consider the transformations $\mathbb{R}^{n_X} \ni W(t) \triangleq \tilde{X}(t)\eu^{-\rho t}$, $\mathbb{R}^{n_X} \ni W^n(t) \triangleq \tilde{X}^n(t)\eu^{-\rho t}$, $\mathbb{R}^{n_z} \ni w(\xi,t) \triangleq \tilde{z}(\xi,t)\eu^{-\rho t}$, and $\mathbb{R}^{n_z} \ni w^n(\xi,t) \triangleq \tilde{z}^n(\xi,t)\eu^{-\rho t}$, with $\rho \in \mathbb{R}_{>0}$ to be determined. Substitution into~\eqref{eq:obSGAII} and~\eqref{eq:obSGAIIn} yields respectively the two ODE-PDE systems:
\begin{subequations}\label{eq:obSGAIInrho1}
\begin{align}
\begin{split}
& \dot{W}(t) = -\rho W(t)+ \bar{A}_1W(t) + G_1(\mathscr{K}_1w)(t)  \\
& \qquad\quad+ F_1(t)(\mathscr{K}_2w)(t)  + f_1(t),\quad  t \in (0,T),
\end{split} \label{eq:SystemsODEObsErrLnrho1}\\
\begin{split}
& \dpd{w(\xi,t)}{t} + \Lambda \dpd{w(\xi,t)}{\xi} = -\rho w(\xi,t) + F_2(t)w(\xi,t)\\
& \qquad \qquad \qquad \qquad \qquad \qquad +f_2(\xi,t), \\
&\qquad \qquad \qquad \qquad \qquad \qquad  (\xi,t) \in (0,1) \times (0,T),
\end{split} \label{eq:originalSystemsPDEObsErrLnrho1} \\
& w(0,t) = 0, \quad t \in (0,T),\label{eq:originalSystemsBCObsErrLnrho1}
\end{align}
\end{subequations}
and
\begin{subequations}\label{eq:obSGAIInrho}
\begin{align}
\begin{split}
& \dot{W}^n(t) = -\rho W^n(t)+ \bar{A}_1W^n(t) + G_1(\mathscr{K}_1w^n)(t)  \\
& \qquad\quad+ F_1^n(t)(\mathscr{K}_2w^n)(t)  + f_1^n(t),\quad  t \in (0,T),
\end{split} \label{eq:SystemsODEObsErrLnrho}\\
\begin{split}
& \dpd{w^n(\xi,t)}{t} + \Lambda \dpd{w^n(\xi,t)}{\xi} = -\rho w^n(\xi,t) + F_2^n(t)w^n(\xi,t)\\
& \qquad \qquad \qquad \qquad \qquad \qquad +f_2^n(\xi,t), \\
&\qquad \qquad \qquad \qquad \qquad \qquad  (\xi,t) \in (0,1) \times (0,T),
\end{split} \label{eq:originalSystemsPDEObsErrLnrho} \\
& w^n(0,t) = 0, \quad t \in (0,T),\label{eq:originalSystemsBCObsErrLnrho}
\end{align}
\end{subequations}
It is straightforward to verify that, for $\rho \in \mathbb{R}_{>0}$ large enough, the unbounded operator $(\mathscr{A}_\rho,\mathscr{D}(\mathscr{A}_\rho))$, defined by
\begin{subequations}
\begin{align}
\bigl(\mathscr{A}_\rho (Y,v)\bigr) & \triangleq \begin{bmatrix} -\rho Y + \bar{A}_1Y + G_1(\mathscr{K}_1v)\\-\Lambda\dpd{v(\xi)}{\xi}-\rho v(\xi)\end{bmatrix}, \\
\mathscr{D}(\mathscr{A}_\rho) & \triangleq \bigl\{(Y,v) \in \altmathcal{Y}\mathrel{\big|} v(0) = 0\bigr\},
\end{align}
\end{subequations}
satisfies $\mathscr{A}_\rho \in \mathscr{G}(\altmathcal{X};1,-\omega_\rho)$, where $\omega_\rho \in \mathbb{R}_{>0}$ may be made arbitrarily large by increasing $\rho$. Denoting by $T_{\mathscr{A}_\rho}(t)$, $t \in [0,T]$, the $C_0$-semigroup generated by $\mathscr{A}_\rho$, in the abstract setting, for all ICs $(W_0,w_0) \triangleq (W(0),w(\cdot,0)) \in \altmathcal{X}$ and $(W_0^n,w_0^n) \triangleq (W^n(0),w^n(\cdot,0)) \in \altmathcal{X}$, the mild solutions of~\eqref{eq:obSGAIInrho1} and~\eqref{eq:obSGAIInrho} read respectively (\cite{Pazy}, Chapter 6)
\begin{subequations}
\begin{align}
\begin{split}
\begin{bmatrix} W(t) \\ w(t)\end{bmatrix}&= T_{\mathscr{A}_\rho}(t)\begin{bmatrix} \tilde{X}_0(t) \\ \tilde{z}_0(t)\end{bmatrix} + \int_0^t T_{\mathscr{A}_\rho}(t-t^\prime)\\
& \quad \times \Biggl( F(t)\begin{bmatrix} W(t^\prime) \\ w(t^\prime)\end{bmatrix} + f(t^\prime) \Biggr)\dif t^\prime, 
\end{split} \label{eq:solMild1} \\
\begin{split}
\begin{bmatrix} W^n(t) \\ w^n(t)\end{bmatrix}&= T_{\mathscr{A}_\rho}(t)\begin{bmatrix} \tilde{X}_0^n(t) \\ \tilde{z}_0^n(t)\end{bmatrix} + \int_0^t T_{\mathscr{A}_\rho}(t-t^\prime)\\
& \quad \times \Biggl( F^n(t)\begin{bmatrix} W^n(t^\prime) \\ w^n(t^\prime)\end{bmatrix} + f^n(t^\prime) \Biggr)\dif t^\prime, \quad t \in [0,T],
\end{split}\label{eq:solMild2}
\end{align}
\end{subequations}
where $ f, f^n : \altmathcal{X} \times [0,T]\to \altmathcal{X}$, with $\mathbb{R}^{n_X+n_z} \ni f(\xi,t) \triangleq [f_1^{\mathrm{T}}(t)\; f_3^{\mathrm{T}}(\xi,t)]^{\mathrm{T}}$ and $\mathbb{R}^{n_X+n_z} \ni f^n(\xi,t) \triangleq [f_1^{n\mathrm{T}}(t)\; f_3^{n\mathrm{T}}(\xi,t)]^{\mathrm{T}}$, and
\begin{align}
F(t) & \triangleq \begin{bmatrix} F_1(t) & 0\\ 0 & F_2(t) \end{bmatrix}, \quad F^n(t) \triangleq \begin{bmatrix} F_1^n(t) & 0\\ 0 & F_2^n(t) \end{bmatrix}.
\end{align}
Defining $\mathbb{R}^{n_X}\ni \tilde{W}^n(t) \triangleq W(t)-W^n(t)$ and $\mathbb{R}^{n_z}\ni \tilde{w}^n(\xi,t) \triangleq w(\xi,t)-w^n(\xi,t)$, subtracting~\eqref{eq:solMild2} from~\eqref{eq:solMild1}, and using the fact that $\mathscr{A}_\rho \in \mathscr{G}(\altmathcal{X};1,-\omega_\rho)$ yields
\begin{align}\label{eq:ineqaaaa}
\begin{split}
& \norm{(\tilde{W}^n(t),\tilde{w}^n(\cdot,t))}_\altmathcal{X} \leq \eu^{-\omega_\rho t}\norm{(\tilde{W}_0^n,\tilde{w}_0^n(\cdot))}_\altmathcal{X} \\
& \quad + \dfrac{\norm{F(\cdot)}_\infty}{\omega_\rho}\norm{(\tilde{W}^n(\cdot),\tilde{w}^n(\cdot,\cdot))}_\infty \\
& \quad + \dfrac{\norm{(W^n(\cdot),w^n(\cdot,\cdot))}_\infty}{\omega_\rho}\norm{F(\cdot)-F^n(\cdot)}_\infty \\
& \quad + \dfrac{1}{\omega_\rho}\norm{f(\cdot,\cdot)-f^n(\cdot,\cdot)}_\infty, \quad t \in [0,T],
\end{split}
\end{align}
where $\altmathcal{X} \ni (\tilde{W}_0^n, \tilde{w}_0^n(\xi)) \triangleq (\tilde{W}^n(0), \tilde{w}^n(\xi,0))$, and
\begin{subequations}
\begin{align}
\norm{(W^n(\cdot),w^n(\cdot,\cdot))}_\infty \triangleq \sup_{t\in \mathbb{R}_{\geq 0}} \norm{(W^n(t),w^n(\cdot,t))}_\altmathcal{X}, \\
\norm{(\tilde{W}^n(\cdot),\tilde{w}^n(\cdot,\cdot))}_\infty \triangleq \sup_{t\in \mathbb{R}_{\geq 0}} \norm{(\tilde{W}^n(t),\tilde{w}^n(\cdot,t))}_\altmathcal{X}.
\end{align}
\end{subequations}
Theorem~\ref{thmObse} ensures that $\norm{(W^n(\cdot),w^n(\cdot,\cdot))}_\infty$ is bounded for $t \in \mathbb{R}_{\geq 0}$, and $\norm{F(\cdot)}_\infty$ exists and is bounded by assumption. Since $\omega_\rho \in \mathbb{R}_{>0}$ may be selected arbitrarily to satisfy $\omega_\rho > \norm{F(\cdot)}_\infty$, the above inequality~\eqref{eq:ineqaaaa} implies the existence of a constant $\beta_\rho \in \mathbb{R}_{>0}$ such that
\begin{align}\label{eq:lslsl}
\begin{split}
& \norm{(\tilde{W}^n(\cdot),\tilde{w}^n(\cdot,\cdot))}_\infty \leq \beta_\rho\biggl( \norm{(\tilde{W}_0^n,\tilde{w}_0^n(\cdot))}_\altmathcal{X} \\
& + \norm{F(\cdot)-F^n(\cdot)}_\infty+\norm{f(\cdot,\cdot)-f^n(\cdot,\cdot)}_\infty\biggr), \quad t \in [0,T].
\end{split}
\end{align}
Letting $n \to \infty$ in~\eqref{eq:lslsl} and recalling~\eqref{eq:Xnnnnnnnn} provides~\eqref{eq:supCobv}.
\end{proof}

\end{document}